\documentclass[journal]{IEEEtran}
\usepackage{amsmath}
\usepackage{graphicx}
\usepackage{amssymb}
\usepackage{booktabs}
\usepackage{mathrsfs}
\usepackage{color}
\usepackage{bm}
\usepackage{caption}
\usepackage{subfig}
\usepackage{enumerate}
\usepackage{multicol}
\usepackage{amsthm}
\usepackage{multirow}
\usepackage{cite}
\usepackage[linesnumbered,ruled]{algorithm2e}
\newtheorem{theorem}{\hspace{0em}Theorem}

\newtheorem{proposition}{\hspace{0em}Proposition}
\newtheorem{remark}{\hspace{0em}Remark}
\newtheorem{example}{\hspace{0em}Example}

\allowdisplaybreaks[4]
\ifCLASSINFOpdf
\else
\fi
\begin{document}
\title{Fusion of labeled RFS densities with \\ minimum information loss}

\author{Lin~Gao, %~\IEEEmembership{Member,~IEEE,}
	Giorgio Battistelli, %~\IEEEmembership{Fellow,~OSA,}
	and~Luigi~Chisci
%	,~\IEEEmembership{Senior~Member,~IEEE}% <-this % stops a space
	\thanks{
		L. Gao, G. Battistelli and L. Chisci are with Dipartimento di Ingegneria dell'Informazione (DINFO),  Universit\`{a} degli Studi di Firenze, Italy.
		E-mail: \{lin.gao,giorgio.battistelli,luigi.chisci\}@unifi.it}}


\markboth{Journal of \LaTeX\ Class Files,~Vol.~14, No.~8, August~2015}%
{Shell \MakeLowercase{\textit{et al.}}: Bare Demo of IEEEtran.cls for IEEE Journals}



\maketitle

\begin{abstract}
%	This paper addresses the problem of multitarget densities fusion (MDF).
	This paper addresses fusion of \textit{labeled random finite set} (LRFS) densities
	according to the criterion of \textit{minimum information loss} (MIL).
%	where multitarget state is modeled with labeled random finite set (LRFS).
%	Specifically,
%	the fusion of multiple local densities provided by local multitarget trackers of agents is considered,
%	with the principle that the fused density leads to minimum weighted information loss (MWIL).
	The MIL criterion amounts to minimizing the (weighted) sum of \textit{Kullback-Leibler divergences} (KLDs)
	with the fused density appearing as righthand argument of the KLDs.
	In order to ensure the fused density to be consistent with the local ones
	when LRFS densities are \textit{marginal $\delta$-generalized labeled multi-Bernoulli} (M$\delta$-GLMB) or \textit{labeled multi-Bernoulli} (LMB) densities,
	the MIL rule is further elaborated by imposing the constraint 
	that the fused density be in the same family of local ones.
	In order to deal with different \textit{fields-of-view} (FoVs) of the local densities,
	the global label space is divided into disjoint subspaces
	which represent the exclusive FoVs and the common FoV of the agents,
	and each local density is decomposed into the sub-densities defined in the corresponding subspaces. 
	%and therefore, independently to each other,
	%by minimizing the KLD from the re-constructed local density (by multiplying sub-densities) to the original one.
	Then fusion is performed subspace-by-subspace to combine local sub-densities into global ones,
	and the global density is obtained by multiplying the global sub-densities.
	Further, in order to tackle the label mismatching issue
	arising in practical applications,
	a \textit{rank assignment optimization} (RAO) of a suitably defined cost is carried out so as to match labels
	from different agents.
%	performed in which 
%	the cost is defined as the information-theoretic divergence between the spatial PDFs among tracks.
%        As a result, the labels are matched by solving the defined RAO.
	Moreover, issues concerning implementation of the MIL rule and
	its application to \textit{distributed multitarget tracking} (DMT) are discussed.
	Finally, the performance of the proposed fusion approach is assessed via simulation experiments
	considering DMT with either the same or different FoVs of the agents.
\end{abstract}


% \fi
\IEEEpeerreviewmaketitle

\section{Introduction}

\IEEEPARstart{O}{riginating} from \cite{mahler2000optimal},
\textit{generalized covariance intersection} (GCI) has become the most commonly adopted method for 
the fusion of multi-object densities.
As well known, GCI fusion amounts to computing the geometric mean of the  local densities \cite{bailey2012conservative}
and is consistent with the \textit{logarithmic opinion pool} (LogOP) \cite{genest1986combining},
which aims to aggregate information from multiple \textit{probability density functions} (PDFs).
Relying on the GCI approach,
several algorithms have been developed for fusing different types of \textit{random finite set} (RFS) processes \cite{uney2013distributed,battistelli2013consensus,wang2017distributed,fantacci2018robust,li2018robust,li2019computationally}.
It has been shown that,
based on the \textit{principle of minimum discrimination of information} (PMDI),
the GCI-fused density 
is the one that minimizes the weighted sum of \textit{Kullback-Leibler divergences}  (KLDs)  from the local densities to itself \cite{battistelli2013consensus,battistelli2014kullback}
and, from an information-theoretic viewpoint, can  be interpreted as the one that leads to 
\textit{minimum information gain} (MIG) \cite{gao2019multiobject,shore1980axiomatic}.

Besides GCI fusion,
it is possible to exploit the dual fusion rule that leads to \textit{minimum information loss} (MIL) \cite{gao2019multiobject,roy1982minimizing}.
Such fusion rule can be obtained also based on the idea of PMDI,
where the fused density is defined as the one minimizing the weighted sum of KLDs 
from itself to the local densities.
It has been shown in \cite{abbas2009kullback} that the fused density according to MIL  
turns out to be the weighted arithmetic mean of the local densities,
which is consistent with the \textit{linear opinion pool} (LOP) \cite{genest1986combining}.
%and is has been named the arithmetic mean (AM) in some of previous literature \cite{bailey2012conservative,li2018partial,li2019local}.
However, the MIL rule cannot be directly applied to fuse the majority of RFS densities due to lack of closeness, i.e.
the resulting fused RFS density does not in general belong to the same family of the local ones.
This prevents, for instance, its direct use in \textit{distributed multitarget tracking} (DMT) \cite{blair2000multitarget}
wherein the fused density at a given time serves as prior information for the next recursion.
%recursively performed applications, e.g. distributed multitarget tracking (DMT) \cite{blair2000multitarget}.
%next iteration of local filtering.
In order to overcome such difficulties,
it is proposed to approximate the fused RFS with a \textit{multi-object Poisson process} (MPP)
matching the first-order statistical moment,
which results into the so-called \textit{arithmetic fusion} \cite{li2018partial,li2019local}.
It has been shown in \cite{gostar2017cauchy} that
such approximation turns out to be the one that minimizes the average \textit{Cauchy-Schwarz divergence} (CSD) \cite{hoang2015cauchy}.
However, all the methods of \cite{gostar2017cauchy,li2018partial,li2019local} can only be applied to the case where local densities are MPP.
In \cite{gao2019multiobject},
by further exploiting the MIL paradigm,
a constraint that 
the fused density must be within the same family of the local ones
is imposed to the PMDI,
so that the ``best", in the sense of MIL, density within the considered family is obtained,
and such result can be applied to general multi-object processes (i.e., i.i.d. cluster processes).

It has been shown that both GCI and MIL fusion rules are conservative and immune to the problem of double counting of information \cite{bailey2012conservative,battistelli2015distributed}.
Moreover, both of them have their respective advantages and disadvantages.
The GCI rule has been proved to guarantee stability in terms of mean-square boundnedess of the estimation error
in the context of distributed state estimation (i.e. distributed Kalman filtering) \cite{battistelli2014kullback,battistelli2015consensus,battistelli2016stability}.
However, the GCI rule suffers from cardinality inconsistency in the context of multi-object density fusion \cite{uney2019fusion},
and is sensitive to misdetections.
Conversely, the MIL rule has satisfactory performance in terms of cardinality estimation,
while its performance deteriorates with higher false alarm rates \cite{gao2019multiobject}.
To summarize, 
in the context of DMT
it is more suitable to adopt MIL fusion whenever the detection probability is low,
while GCI fusion is preferable whenever dense clutter is present in the \textit{area of interest} (AoI).

In this paper,
the primary concern is in the extension of MIL multi-object fusion
to  \textit{labeled RFS} (LRFS) densities.
The main advantage of modeling the multi-object state as LRFS is that
the trajectory of each object can be obtained directly,
while additional track management procedures  \cite{panta2009data} are needed to extract object trajectories from unlabelled RFS densities.
It has been shown in \cite{papi2015generalized} that
a general LRFS density can be factored into the product of the
\textit{joint existence probability} (JEP) of the multi-object label set
by the corresponding \textit{conditional joint PDF} (CJPDF).
%i.e. the spatial distribution conditional on label set.
Based on such representation,
it is shown in this paper that the fusion of general LRFS densities 
(defined on the same label space)
adopting the MIL rule yields another general LRFS density,
and such result can be directly applied to fuse multiple $\delta$-GLMB densities \cite{vo2013labeled,vo2014labeled}.
However,
when the local LRFSs are modeled as  M$\delta$-GLMB \cite{fantacci2015marginalized} 
or LMB \cite{reuter2014labeled} processes,
the resulting fused LRFS density is not of the same type of the local ones.
Then the idea of \cite{gao2019multiobject},
where MIL optimization is restricted to the considered specific class of  local RFS densities,  
is exploited;
specifically, the ``best", in the MIL sense, fused M$\delta$-GLMB/LMB  of local M$\delta$-GLMB/LMB  densities is found.  

In practice,
due to the limitation of sensor range, it turns out that
the multi-object densities to be fused carry information on
different \textit{fields-of-view} (FoVs),
thus implying another challenge of multi-object fusion.
In such situation, if the GCI fusion is directly applied, 
due to its multiplicative nature,
the fused density tends to become null outside the common FoV.
In this way,
the non-common (exclusive) information carried by local densities is lost. 
Such a problem can be alleviated by taking remedies on the GCI method.
Specific remedies are the following.
\begin{itemize}
	\item[-] When multi-object densities are modeled as MPP with \textit{Gaussian mixture} (GM) representation\cite{vo2006gaussian},
	a uniform initialization of the \textit{probability hypothesis density} (PHD) for local MPPs can be employed so as to avoid the null-PHD problem \cite{battistelli2017random}.
	It is also possible to first disengage the Gaussian components (GCs) outside the common FoV
	with component matching algorithms (e.g. clustering algorithm),
	and then separately perform fusion on the GCs inside and outside the common FoV 
	with different strategies \cite{vasic2016system,li2019distributed}.
	\item[-] When the LMB RFS \cite{reuter2014labeled} is employed to model the multi-object state,
	a promising strategy is to associate to each \textit{Bernoulli component} (BC) of each local LMB density a specific fusion weight 
	based on the amount of information it carried,
	and then the fusion weights of BCs
	which have not been updated by measurements are automatically decreased,
	thus reducing their effect on the fusion process \cite{wang2018centralized}.
	Moreover, motivated from the uniform initialization strategy in \cite{battistelli2017random}, it is also possible to adopt a density compensation strategy, 
	where the local posterior of each agent undergoes an auxiliary birth process outside its local FoV.
	As  a result, the problem of miss-detections outside the local FoV of each agent can be alleviated \cite{li2018multi}.
\end{itemize}

Unlike GCI fusion which essentially performs ``intersection" among the agent FoVs,
the MIL rule has the potential to correctly fuse multi-object densities defined in different FoVs \cite{gostar2017cauchy,gao2019multiobject}.
However, since each local density has only the information within its own FoV,
the JEPs of all label subsets  that include targets outside the FoV
are always zero.
If the MIL rule is directly applied to all the local densities,
the JEP of the global density that includes all the targets spread over the whole surveillance area 
will certainly become null,
which means that it is not possible to jointly detect all the  existing targets spread over the whole surveillance area.

In this paper,
we propose to handle fusion of multi-object densities with different FoVs by applying
the MIL rule to mutually disjoint label subspaces,
where the label subspaces are obtained by evaluating the exclusive and common FoVs of the agents.
The sub-densities,
which are defined on different label subspaces,
are found by minimizing the KLD 
from the re-constructed local density (equal to the product of sub-densities) 
and the corresponding original local density.
By combination of the MIL rule and the decomposition strategy,
the problem of fusing local LRFS densities defined on different FoVs can be handled.
The advantage of the proposed algorithm is that
it needs neither aforehand initialization of multi-object densities over the global FoV 
nor density compensations,
so that it can be implemented in a more efficient way.

It should be noted that
the proposed MIL fusion rule for multiple LRFS densities is based on the pre-condition
that all the involved LRFS densities are defined on the same label space.
In practice, however,
%since each agent works independently to each other in the context of DMT,
it is extremely difficult to ensure such assumption
due to the fact that the local LRFS densities are propagated independently,
thus resulting into the \textit{label mismatching} (LM) problem \cite{li2019computationally}.
This difficulty can be overcome by setting up associations among labels of different LRFS densities.
The existing strategy \cite{li2019computationally} exploits rank assignment  to find the associations,
in which each label of the LRFS density with smaller cardinality of the label space 
will always be associated to a label of another LRFS density.
Such strategy works well when all agents have the same FoV,
nevertheless, whenever agents have different FoVs,
it is also possible that some label of an LRFS density remains unassociated,
thus the method in \cite{li2019computationally} is not suitable.
In this paper, 
we propose to solve the LM problem with different FoVs
also by means of the rank assignment problem,
where the cost is defined by exploiting an information-theoretic divergence between BCs.
In the proposed strategy,
the cost that a BC remains unassociated is also defined 
(which actually represents an upper bound on the divergence between associated BCs in different local LRFS densities),
thus the BCs outside the common FoV can be properly found.
%Though such strategy is heuristic,
%it indeed solves the LM problem in the context of DMT with different FoVs.

To summarize, 
this paper provides the following main contributions.
\begin{enumerate}[1)]
	\item A novel fusion rule that leads to MIL is proposed to fuse LRFS densities.
	\item In combination with a suitable label decomposition strategy, the MIL rule can be directly applied to handle  DMT when the agents have different FoVs.
	\item A strategy is proposed to solve the LM problem among LRFS densities, thus strengthening the applicability of the proposed algorithms to real scenarios.
\end{enumerate}
% \textcolor{blue}
% {
%The rest of this paper is organized as follows....
%}

\subsection*{Notation}
%This paper involves numerous of notations,
%n order to assist understanding the contents of this paper,
The notation used throughout the paper  is summarized hereafter.
First, we denote the agent set of a multi-agent system (MAS) as $\cal N$,
which consists of $|{\cal N}|$ agents.
Next, all the quantities related to LRFSs will be denoted with  boldface symbols.
Specifically, we use ${\bf X}$ to denote an LRFS and
${\bf x}$ for the augmented (labelled) single-object state.
Further, ${\boldsymbol \pi}$ represents a generic LRFS density,
${\boldsymbol \pi}_{\delta}$ a $\delta$-GLMB density,
${\boldsymbol \pi}_M$ an M$\delta$-GLMB density,
and ${\boldsymbol \pi}_{\beta}$ an LMB density.
Moreover,
a superscript is used to refer to a specific agent, i.e., ${\boldsymbol \pi}^i$ indicates the local density of agent $i \in {\cal N}$.
Conversely, subscripts of sets will be used to indicate their cardinality.
For instance,
${\bf X}_n$ and $L_n$ denote respectively an LRFS and label set with cardinality $n$.
We also define
$L_n\buildrel \Delta \over= \left\{ {{l_1}, \ldots ,{l_n}} \right\}$
and $X_n \buildrel \Delta \over= \left\{ {{x_1}, \ldots ,{x_n}} \right\}$.
For the sake of convenience, 
in the rest of this paper,
the symbols $L_n$, $X_n$ and their respective full definitions $\left\{ {{l_1}, \ldots ,{l_n}} \right\}$,
$\left\{ {{x_1}, \ldots ,{x_n}} \right\}$
will be interchangeably used.
All the involved  spaces will be denoted by blackboard bold symbols.
For instance, $\mathbb X$ denotes the state space,
and $\mathbb L$ the label space.
Further, we use subscripts with space symbols to refer to  subspaces,
e.g. $\mathbb L = \mathop  \uplus \nolimits_{m = 1}^M {\mathbb L_m}$,
where $\uplus$ denotes disjoint union (i.e. $\mathbb L_m \cap \mathbb L_{m'} = \emptyset$, for $m \ne m'$).
Conversely,
we use superscripts together with space symbols to refer to the label space of a local LRFS density,
i.e. $\mathbb L^i$ indicates the label space of ${\boldsymbol \pi}^i$, for agent $i\in {\cal N}$.
Finally, we define ${\cal F}_n(\mathbb L)$ as the set of all subsets of $\mathbb L$ with $n$ elements.

\section{Background}
\subsection{Labeled RFS}
In this paper, the multi-object state ${\bf X}_n = \{{\bf x}_1,\ldots,{\bf x}_n\}$ with  cardinality $n$
is modeled as an LRFS
in which the $k$-th ($k = 1,\ldots,n$) single-object state is denoted as ${\bf x}_{k} = \left( {x_k,l_k} \right) \in {{\mathbb X}} \times {\mathbb L}$,
${\mathbb X}$ denoting the kinematic state space
and ${\mathbb L}$ the label space. 
From a statistical viewpoint \cite{mahler2014advances},
an LRFS is completely characterized by its multi-object density ${\boldsymbol \pi}$.
%i.e. the multitarget state $\bf X$ can be extracted from ${\boldsymbol \pi}$.
%hence the problem of multitarget tracking can be re-casted to recursively compute ${\boldsymbol \pi}_t$ at every time $t$,
%and then the estimated multitarget state $\widehat {\bf X}_t$ is extracted from ${\boldsymbol \pi}_t$.
For a general LRFS density ${\boldsymbol \pi}$, 
its \textit{joint existence probability} (JEP) $p$ of label set $L_n\buildrel \Delta \over= \left\{ {{l_1}, \ldots ,{l_n}} \right\}$ is given by \cite{papi2015generalized}
\begin{align}
p\left(L_n\right)= \int { \ldots \int {{\bf{\pi }}\left( {\left\{ {\left( {{x_1},{l_1}} \right), \ldots ,\left( {{x_n},{l_n}} \right)} \right\}} \right)d{x_1} \cdots d{x_n}} }.    \label{eq:PMFL}
\end{align}
Then, it is straightforward to define the 
\textit{conditional joint probability density function} (CJPDF) $f$ of RFS $X_n \buildrel \Delta \over= \left\{ {{x_1}, \ldots ,{x_n}} \right\}$ 
given label set $L_n$ as \cite{papi2015generalized}
\begin{align}
f\left( {\left\{ {\left( {{x_1}|{l_1}} \right), \ldots ,\left( {{x_n}|{l_n}} \right)} \right\}} \right) \buildrel \Delta \over = \frac{{{\boldsymbol \pi}\left( {\left\{ {\left( {{x_1},{l_1}} \right), \ldots ,\left( {{x_n},{l_n}} \right)} \right\}} \right)}}{{p\left( L_n \right)}}.  \label{eq:CJPD}
\end{align}
It can be directly seen from the definition (\ref{eq:CJPD}) that
the CJPDF $f$ is permutation-invariant,
i.e.
\begin{align}
& f\left( {\left\{ {\left( {{x_1},{l_1}} \right), \ldots ,\left( {{x_n},{l_n}} \right)} \right\}} \right) = \nonumber \\
&\quad\quad\quad\quad f\left( {\left\{ {\left( {{x_{\sigma_n \left( 1 \right)}},{l_{\sigma_n \left( 1 \right)}}} \right), \ldots ,\left( {{x_{\sigma_n \left( n \right)}},{l_{\sigma_n \left( n \right)}}} \right)} \right\}} \right),   \label{eq:PI}
\end{align}
where $\sigma_n$ denotes any permutation on numbers $1,\ldots,n$, and $\sigma_n(i)$ its $i$-th element ($i=1,\ldots,n$).

For the sake of convenience, we introduce the shorthand notation $f\left( {\left. X_n \right|L_n} \right) \buildrel \Delta \over = f\left( {\left\{ {\left( {{x_1}|{l_1}} \right), \ldots ,\left( {{x_n}|{l_n}} \right)} \right\}} \right)$.
Equivalently, any LRFS density ${\boldsymbol \pi}$ can be generally expressed as
\begin{align}
{\boldsymbol \pi}\left( {\bf X}_n \right) = p\left( L_n \right) \cdot f\left( X_n|L_n \right).   \label{eq:LRFSN}
\end{align}

Hence,
any LRFS density can be completely specified by the JEP $p$ and CJPDF $f$ 
according to (\ref{eq:LRFSN}).
In particular,
\begin{itemize}
	\item[-] A {\bf $\boldsymbol \delta$-GLMB density} ${\boldsymbol \pi}_{\delta}=(p_{\delta},f_{\delta})$ is specified by \cite{vo2013labeled}
	\begin{align}
	p_{\delta}\left( L_n \right) &= \sum\limits_{\xi \in \Xi} {{w^\xi}\left( L_n \right)}, \label{eq:GLMBpmf} \\
	f_{\delta}\left( {\left. X_n \right|L_n} \right) &= \sum\limits_{\xi \in \Xi} {\frac{{{w^\xi}\left( L_n \right)}}{{\sum\limits_{\xi' \in \Xi} {{w^{\xi'}}\left( L_n \right)} }}\prod\limits_{k = 1}^n {f_{\left. {{l_k}} \right|L_n}^\xi\left( {{x_k}} \right)} }, \label{eq:GLMBspdf}
	\end{align}
	where: $\Xi$ is a discrete index set whose elements represent
	%which can be specified to the
	 track-to-measurement association hypotheses in the context of multitarget tracking with point measurements;
	${{w^\xi}\left( L_n \right)}$ denotes the JEP of $L_n$ under hypothesis $\xi$ 
	which satisfies $\sum\nolimits_{L \subseteq {\mathbb L}} {\sum\nolimits_{\xi \in \Xi} {{w^\xi}\left( L \right)} } = 1$;
	${f_{\left. {{l_k}} \right|L_n}^\xi}$ represents the  PDF of track $l_k$ 
	conditional on $L_n$ and hypothesis $\xi$;
	\item[-] An {\bf M$\boldsymbol \delta$-GLMB density} ${\boldsymbol \pi}_{M}=(p_{M},f_{M})$, 
	which is defined as $\delta$-GLMB density marginalized 
	by the discrete index set $\Xi$,
        is specified by  \cite{fantacci2015marginalized}
	\begin{align}
	p_{M}\left( L_n \right) &= w \left( L_n \right), \label{eq:MGLMBpmf}  \\
	f_{M}\left( {\left. X_n \right|L_n} \right) &= \prod\limits_{k = 1}^n {{f_{\left. {{l_k}} \right|L_n}}\left( {{x_k}} \right)},   \label{eq:MGLMBspdf}
	\end{align}
	where $w\left( L_n \right)$ denotes the JEP of label set $L_n$
	and ${f_{\left. {{l_k}} \right|L_n}}$ the PDF of track $l_k$ 
	conditional on label set $L_n$;
%	It can also be seen from (\ref{eq:MGLMBspdf}) and (\ref{eq:GLMBspdf}) 
	\item[-] An {\bf LMB density} ${\boldsymbol \pi}_{\beta}=(p_{\beta},f_{\beta})$
	is specified by \cite{reuter2014labeled}
	\begin{align}
	p_{\beta}\left( L_n \right) &= \prod\limits_{l \in {\mathbb L} } {\left( {1 - {r_l}} \right)} \prod\limits_{l' \in L_n} \frac{r_{l'}}{1 - r_{l'}}  \label{eq:LMBwl} ,  \\
	f_{\beta}\left( {\left. X_n \right|L_n} \right) &= \prod\limits_{l \in L_n} {{f_{l}}\left( x \right)}  \label{eq:LMBspdf} ,
	\end{align}
	where $r_l$ denotes the \textit{existence probability} (EP) of track with label $l$
	and $f_l$ the corresponding PDF.
\end{itemize}

\begin{remark}  \label{rem:LCILRFS}
	Besides the above mentioned definition as marginalization with respect to $\Xi$ of the $\delta$-GLMB density ${\boldsymbol \pi}_{\delta}$ \cite{fantacci2015marginalized},
	an M$\delta$-GLMB density can also be defined in a more general manner.
	As indicated by (\ref{eq:MGLMBpmf}) and (\ref{eq:MGLMBspdf}),
	an M$\delta$-GLMB density ${\boldsymbol \pi}_{M}=(p_{M},f_{M})$ can be re-defined as the LRFS density 
	given by (\ref{eq:LRFSN})
	with CJPDF $f_M$ independent of the PDF of each track conditionally on the track set. \\
%	Hence, in the remaining part of this paper, 
%	the M$\delta$-GLMB density is also referred to as the CJPD conditional independent LRFS (CCI-LRFS) density. 
%\textcolor{blue}{Luigi: Possibly remove since it is not very clear and not so important in the context of DMT.}
\end{remark}

%\begin{remark}  \label{rem:Com}
%	It is more accurate to model the multitarget state as a $\delta$-GLMB density than the M$\delta$-GLMB and LMB densities.
%	Besides,
%	$\delta$-GLMB density is a conjugate prior under point measurement models \cite{vo2013labeled},
%	while the Bayesian update of M$\delta$-GLMB/LMB density under point measurement models
%	will become $\delta$-GLMB, and approximations must be performed 
%	to convert the $\delta$-GLMB density back to M$\delta$-GLMB/LMB density
%	in order to provide the prior information of next iteration in single sensor multitarget tracking.
%	However, the number of hypotheses of $\delta$-GLMB density increases exponentially along with time 
%	if no additional operation (i.e. pruning of hypothesis) is enrolled,
%	which takes a huge amount of memory source as well as computational load,
%	thus is not engineering friendly.
%	In this regard, it is practically more desired to adopt M$\delta$-GLMB and LMB trackers 
%	for multitarget tracking.
%\end{remark}

\begin{remark}
	It can be seen from (\ref{eq:LMBwl}) and (\ref{eq:LMBspdf}) that,
	compared to an  M$\delta$-GLMB density,
	the JEP of an LMB density is further assumed to be independent of the EPs of the involved labels.
	Further, it can be concluded that
	the LMB density is also completely charactered by the existence probability (EP) $r_l$ and PDF $f_l$ of each track $l\in{\mathbb L}$.
	Hence, we also introduce the shorthand notation ${\boldsymbol \pi}_{\beta} = \{(r_l,f_l)\}_{l\in{\mathbb L}}$ for an LMB density.
\end{remark}

\subsection{Fusion with GCI}
%Denote the agent set of a \textit{multi-agent system} (MAS) as $\cal N$,
%which consists of $|{\cal N}|$ agents.
In this paper, it is assumed that
each agent $i \in {\cal N}$ has the ability to 
compute a local density ${\boldsymbol \pi}^i$ with measurements provided by sensors onboard
and also to transmit and receive data.
The goal of fusion amounts to
compute the global density $\overline{\boldsymbol \pi}$ that encapsulates all the information
provided by local ones ${\boldsymbol \pi}^i, i \in {\cal N}$.
%In this section, the existing fusion strategy is briefly introduced.
%Since MDF is supposed to be performed at every iteration,
%the time symbol is omitted hereafter.
So far, the most commonly adopted fusion strategy for LRFS densities is the so called
\textit{generalized covariance intersection} (GCI) \cite{mahler2000optimal} (also known as logarithmic opinion pool \cite{genest1986combining}) 
%has shown its effectiveness in fusion of LRFS densities, 
according to which the global posterior $\overline{\boldsymbol \pi}_{\rm GCI}$ is given by
\begin{align}
\overline{\boldsymbol \pi}_{\rm GCI} \left( {\bf{X}} \right) = \frac{{\prod\limits_{i \in {\cal N}} {{{\left[ {{{\boldsymbol \pi} ^i}\left( {\bf{X}} \right)} \right]}^{{\omega ^i}}}} }}{{\int {\prod\limits_{i \in {\cal N}} {{{\left[ {{{\boldsymbol \pi} ^i}\left( {\bf{X}} \right)} \right]}^{{\omega ^i}}}} \delta {\bf{X}}} }} ,  \label{eq:GCI}
\end{align}
where $\omega^i$ are suitable non-negative weights summing up to
unity,
and the involved integral is defined with respect to LRFSs, see \cite[Proposition 2]{vo2013labeled}. 
Based on such a fusion rule,
the global LRFS density can be explicitly computed when the multi-object state is modeled by either an M$\delta$-GLMB or LMB process.

Recently it has been pointed out that the fused density $ {\boldsymbol \pi}$ computed by the GCI rule
turns out to be the \textit{weighted Kullback-Leibler average} (wKLA) \cite{battistelli2014kullback,battistelli2015distributed} defined as follows
\begin{align}
\overline{\boldsymbol \pi}_{\rm GCI}  \buildrel \Delta \over = \arg \mathop {\min }\limits_{\boldsymbol \pi}  \sum\limits_{i \in {\cal N}} {{D_{\rm KL}}\left( {\left. {\boldsymbol \pi} \right\|{{\boldsymbol \pi} ^i}} \right)} ,  \label{eq:WKLA}
\end{align}
where ${D_{\rm KL}}\left( {\left. {{{\boldsymbol \pi}^1}} \right\|{{\boldsymbol \pi}^2}} \right)$ is the Kullback-Leibler divergence (KLD) from ${\boldsymbol \pi}^2$ to ${\boldsymbol \pi}^1$ defined as 
\begin{align}
{D_{\rm KL}}\left( {\left. {{{\boldsymbol \pi}^1}} \right\|{{\boldsymbol \pi}^2}} \right) \buildrel \Delta \over = \int {{{\boldsymbol \pi}^1}\left( {\bf X} \right)\log \frac{{{{\boldsymbol \pi}^1}\left( {\bf X} \right)}}{{{{\boldsymbol \pi}^2}\left( {\bf X} \right)}}\delta {\bf X}} . \label{eq:KLD}
\end{align}
From the viewpoint of information theory,
the KLD from ${\boldsymbol \pi}^2$ to ${\boldsymbol \pi}^1$ (i.e. $D_{\rm KL}\left( {\left. {{{\boldsymbol \pi}^1}} \right\|{{\boldsymbol \pi}^2}} \right)$) represents 
the information gain when ${\boldsymbol \pi}^2$ is replaced by ${\boldsymbol \pi}^1$
or, equivalently,
the information loss when ${\boldsymbol \pi}^1$ is replaced by ${\boldsymbol \pi}^2$ \cite{kullback1997information}.
Hence, the GCI rule (\ref{eq:GCI}) is actually the one that results into the \textit{minimum information gain} (MIG)
after fusion  \cite{gao2019multiobject}.

\section{Fusion of LRFS densities with MIL}  \label{sec:MWIL}

\subsection{MIL fusion of LRFS densities}
In this paper, we propose to fuse the local densities by adopting the criterion 
that the global density ${\boldsymbol \pi}$ leads to \textit{minimum information loss} (MIL).
Such fusion rule is defined as follows \cite{gao2019multiobject}
\begin{align}
\overline{\boldsymbol \pi}_{\rm MIL}  = \arg \mathop {\min }\limits_{\boldsymbol \pi}  \sum\limits_{i \in {\cal N}} {{\omega ^i}{D_{\rm KL}}\left( {\left. {{{\boldsymbol \pi} ^i}} \right\|{\boldsymbol \pi} } \right)} ,   \label{eq:MWIL}
\end{align}
whose difference with respect to the MIL criterion 
merely lies in the ordering of arguments, i.e.  local densities ${\boldsymbol \pi}^i$ and the global one ${\boldsymbol \pi}$, in the KLDs.
Since the main concern of this paper is the MIL fusion rule,
from now on we set $\overline{\boldsymbol \pi} \buildrel \Delta \over =\overline{\boldsymbol \pi}_{\rm MIL}$.
The resulting global density $\overline{\boldsymbol \pi}$ is given by
\begin{align}
\overline{\boldsymbol \pi} \left( {\bf{X}} \right) = \sum\limits_{i \in {\cal N}} {{\omega ^i}{{\boldsymbol \pi} ^i}\left( {\bf{X}} \right)} .   \label{eq:GP}
\end{align}

Compared to the GCI criterion, 
fusion with MIL has the advantage of faster detection of newly appeared targets,
while GCI has better performance in rejecting false alarms.
%In this paper,
%the performance of MWIL fusion under different FoVs of agents will be further investigated,
%where the multitarget state is modeled as LRFS.
It has been shown in \cite{gao2019multiobject} that,
for most types of unlabeled RFS multi-object densities,
the fused density computed by %directly adopting the result of 
(\ref{eq:GP}) 
no longer belongs to the same family of local densities,
thus hindering its application to scenarios which
require the conjugacy between local densities and the fused density
(e.g. in the context of DMT).
%hence the fused density can not be utilized as the prior for local filtering of next recursion.
However, such rule can be directly applied to fuse LRFS densities in the general form of (\ref{eq:LRFSN}), 
as shown in the following proposition.
Please notice that it is temporarily assumed
in this section that the labels of all considered LRFS densities have been perfectly matched.
Solving the problem of label mismatching is deferred to Section \ref{sec:LM}.
\begin{proposition}   \label{pro:P0}
	If the local density ${\boldsymbol \pi}^i=(p^i,f^i)$ of each agent $i\in{\cal N}$ 
	is in the form (\ref{eq:LRFSN}),
	and all the local densities are defined on the same label space,
	then the optimal fused LRFS density leading to MIL has density $\overline{\boldsymbol \pi}=(\overline p,\overline f)$
	with JEP $\overline p$ and CJPDF $\overline f$ given by
	\begin{align}
	\overline p\left( L \right) &= {\sum\limits_{i \in {\cal N}} {{\omega ^i}{p^i}\left( L \right)} }, \\
	\overline f\left( {\left. X \right|L} \right) &= \sum\limits_{i \in {\cal N}} {\frac{{{\omega ^i}{p^i}\left( L \right)}}{{\sum\limits_{j \in {\cal N}} {{\omega ^j}{p^j}\left( L \right)} }}{f^i}\left( {\left. X \right|L} \right)} .
	\end{align}
\end{proposition}
Proof: see Appendix \ref{app:P0}.

Proposition \ref{pro:P0} can be directly applied to fuse multiple $\delta$-GLMB densities,
as shown in the following theorem.
\begin{theorem}  \label{the:T1}
	If the local density ${\boldsymbol \pi}_{\delta}^i=(p_{\delta}^i,f_{\delta}^i)$ of each agent $i\in{\cal N}$ 
	is $\delta$-GLMB with discrete index set $\Xi^i$,
	and all the local densities are defined on the same label space,
	then the optimal fused LRFS density leading to MIL has density $\overline {\boldsymbol \pi}_{\delta}=({\overline p}_{\delta},{\overline f}_{\delta})$
	with JEP $\overline p_{\delta}$ and CJPDF $\overline f_{\delta}$ given as follows
	\begin{align}
	{\overline p_\delta }\left( L_n \right) &= \sum\limits_{i \in {\cal N}} {\sum\limits_{\xi \in {\Xi^i}} {{w^{\xi,i}}\left( L_n \right)} },   \\
	{\overline f_\delta }\left( {\left. X_n \right|L_n} \right) &= \sum\limits_{i \in {\cal N}} {\sum\limits_{\xi \in {\Xi^i}} {\frac{{{w^{\xi,i}}\left( L_n \right)}}{{{\overline p_\delta }\left( L_n \right)}}\prod\limits_{k = 1}^n {f_{\left. {{l_k}} \right|L_n}^{\xi,i}\left( {x_k} \right)} } } .
	\end{align}
\end{theorem}
Since the proof of Theorem \ref{the:T1} is quite straightforward from Proposition \ref{pro:P0}, 
it is omitted.
%Concerning the fusion of $\delta$-GLMB densities,
%the following comments should be addressed.
However,
unlike  $\delta$-GLMB densities that are closed under MIL fusion,
fusion of M$\delta$-GLMB/LMB densities by (\ref{eq:GP})
will not result into an M$\delta$-GLMB/LMB density again,
as can be straightforwardly seen.
Hence, labelled multi-object densities encounter the same difficulties in the application of MIL fusion as their unlabelled counterparts.
In this paper,
it is proposed to find the ``best" global M$\delta$-GLMB/LMB density yielding MIL
by explicitly adding the constraint that the solution of (\ref{eq:MWIL}) is of the same type of the fusing densities ${\boldsymbol \pi}^i$,
which is essentially the same idea of applying the MIL rule to 
fuse MPPs and i.i.d. cluster processes in \cite{gao2019multiobject}.
First, we consider the problem of fusing multiple M$\delta$-GLMB densities under the MIL criterion, 
which can be solved by means of the following proposition. 

\begin{proposition}   \label{pro:P1}
	If the local densities ${\boldsymbol \pi}^i_M$, $i\in{\cal N}$, are M$\delta$-GLMB
	with JEP $p^i_M$ and CJPDF $f^i_M$ given by
	\begin{align}
	p^i_{M}\left( L_n \right) &= w^i\left( L_n \right), \\
	f^i_{M}\left( {\left. X_n \right|L_n} \right) &= \prod\limits_{k = 1}^n {{f_{\left. {{l_k}} \right|L_n}^i}\left( {{x_k}} \right)},  
	\end{align}
	and all the local densities are defined on the same label space,
	then the best M$\delta$-GLMB density $\overline{\boldsymbol \pi}_M = (\overline p_M,\overline f_M)$ 
	leading to MIL	is given by
	\begin{align}
	\overline p_M\left( L_n \right) &= \sum\limits_{i \in {\cal N}} {{\omega ^i}{p^i_M}\left( L_n \right)} , \label{eq:Fp} \\
	\overline f_M\left( {\left. X_n \right|L_n} \right) &= \prod\limits_{k = 1}^n {{\overline f_{{l_k}|L_n}}\left( {{x_k}} \right)} ,   \label{eq:Ff}
	\end{align}
	where
	\begin{align}
	{\overline f_{{l_k}|L_n}}\left( {{x_k}} \right) &= \sum\limits_{i \in {\cal N}} {{{\tilde \omega }^i}(L_n)\cdot f_{{l_k}|L_n}^i\left( {{x_k}} \right)}, \quad k = 1,\ldots,n ,  \label{eq:Ffl}   \\
	{{\tilde \omega }^i}(L_n) &= \frac{{{\omega ^i}p_M^i\left( L_n \right)}}{{\sum\nolimits_{j \in {\cal N}} {{\omega ^j}p_M^j\left( L_n \right)} }}.   \label{eq:OmegaMGLMB}
	\end{align}
\end{proposition}
Proof: see Appendix \ref{app:P1}.

%The results of Proposition \ref{pro:P1} can be directly adopted to fuse multiple M$\delta$-GLMB densities,
%however, it can not be applied for fusion of LMB densities,
%since the fused label set PMF (\ref{eq:Fp}) can not be further simplified,
%and hence the global density will become M$\delta$-GLMB.
Next, in order to find the fused LMB density leading to MIL,
the structure of JEP (\ref{eq:LMBwl}) of an LMB density should be further exploited, as shown in the following proposition.
\begin{proposition}  \label{pro:P2}
	If the local density of each agent $i\in{\cal N}$ is modeled as LMB ${\boldsymbol \pi ^i_{\beta}} = {\left\{ {\left( {{r_l^i},{f_l^i}} \right)} \right\}_{l \in {\mathbb L} }}$,
	and all the local densities are defined on the same label space $\mathbb L$,
	then the best LMB density leading to MIL has density $\overline{\boldsymbol \pi}_{\beta} = {\left\{ \left( {\overline r_l,\overline f_l} \right) \right\}_{l \in {\mathbb L} }}$
	with EP $\overline r_l$ and PDF $\overline f_l$ of each label $l\in{\mathbb L}$ given as follows
	\begin{align}
	\overline r_l &= \sum\limits_{i \in {\cal N}} {{\omega ^i}{r^i_l}} ,  \label{eq:FLMBinten}  \\
	{\overline f_l}\left( x \right) &= \sum\limits_{i \in N} {{{\tilde \omega }^i_l}f_l^i\left( x \right)},  \label{eq:FLMBspdf}
	\end{align}
	where 
	\begin{align}
	{{\tilde \omega }^i}_l = \frac{{{\omega ^i}r_l^i}}{{\sum\nolimits_{j \in {\cal N}} {{\omega ^j}r_l^j} }}   \label{eq:OmegaLMB}
	\end{align}
\end{proposition}
Proof: see Appendix \ref{app:P2}.

\begin{remark}  \label{rem:R1}
	It should be noted that
	it is also possible to directly adopt the result of Proposition \ref{pro:P1} in order to fuse multiple LMB densities.
	Nevertheless,
	the resulting global density will become M$\delta$-GLMB.
	This fact can be seen by comparing the fused JEPs computed by (\ref{eq:Fp}) and (\ref{eq:FLMBinten}),
	where the fused JEP  in (\ref{eq:Fp}) is given by
	\begin{align}
	{\overline p }\left( L \right) &= \sum\limits_{i \in {\cal N}} {{\omega ^i}p_\beta ^i\left( L \right)} \nonumber \\
	& = \sum\limits_{i \in {\cal N}} {{\omega ^i}\left[ {\prod\limits_{l \in {\mathbb L} } {\left( {1 - r_l^i} \right)} \prod\limits_{l' \in L} {\frac{{r_{l'}^i}}{{1 - r_{l'}^i}}} } \right]} ,  \label{eq:JEPP2}
	\end{align}
	and the fused JEP in (\ref{eq:FLMBinten}) is given by
	\begin{align}
	{\overline p }\left( L \right) = \prod\limits_{l \in {\mathbb L}} {\left( {1 - \sum\limits_{i \in {\cal N}} {{\omega ^i}r_l^i} } \right)} \prod\limits_{l' \in L} {\frac{{\sum\nolimits_{i \in {\cal N}} {{\omega ^i}r_{l'}^i} }}{{1 - \sum\nolimits_{i \in {\cal N}} {{\omega ^i}r_{l'}^i} }}} .   \label{eq:JEPP3}
	\end{align}
%	Since the PDFs of labels are independent to each other in LMB density,
%	it can be straightforwardly seen that the CJPDFs of fused densities by (\ref{eq:Ff}) 
	However, the resulting global M$\delta$-GLMB density can be converted to LMB density based on matching the probability hypothesis density (PHD) \cite{reuter2014labeled},
	and the resulting LMB density is consistent to the one computed by Proposition \ref{pro:P2},
	as shown in Appendix \ref{app:rem3}.
	In this regard,
	Proposition \ref{pro:P2} can serve as the principled certification that 
	such conversion  can lead to minimum information loss. 
	Furthermore, 
	the results of Proposition \ref{pro:P2} are also practically valuable.
	Proposition \ref{pro:P2} indicates that
	fusion of multiple LMB densities defined on the same label space
	amounts to performing a label-wise MIL fusion of BCs,
	thus its computational load increases linearly with the number of BCs.
	Instead,
	the fusion of multiple M$\delta$-GLMB amounts to performing label-set-wise MIL fusion,
	and the computational load turns out to increase exponentially with the number of labels.
\end{remark}

\subsection{Accuracy analysis}
It has been pointed out that
the \textit{MIL-optimal fused density} (MIL-OFD) of M$\delta$-GLMB/LMB densities
is no longer  an M$\delta$-GLMB/LMB density,
thus turns out to be practically useless in the context of recursive local multi-object filtering.
In Propositions \ref{pro:P1} and \ref{pro:P2}, it is proposed
to find the best, in the MIL sense, fused density within the same M$\delta$-GLMB/LMB family
of the local densities.
In this respect,
a natural question 
concerns the accuracy of the M$\delta$-GLMB/LMB approximation,  provided by Proposition 2/3, of the MIL-OFD.
Such a question 
 is addressed in the following theorem.
%In this subsection, we address such question 
%on analyzing the KLD from the fused M$\delta$-GLMB/LMB density to OFD with MIL rule.
%Specifically, the accuracy when the fused M$\delta$-GLMB
%is adopted to replace the OFD 
%applying the results of Proposition \ref{pro:P1} 
%is analyzed in the following theorem.
\begin{theorem}  \label{the:T2}
	The KLD from the fused M$\delta$-GLMB/LMB of Proposition \ref{pro:P1}/\ref{pro:P2} to the MIL-OFD is bounded by 
	the average KLD among all pairs of agents,
	i.e. 
	\begin{align}
	{D_{\rm KL}}\left( { {\sum\limits_{i \in {\cal N}} {{\omega ^i}\pi _M^i} } \|{{\overline \pi }_M}} \right) & \le \sum\limits_{i \in {\cal N}} {\sum\limits_{j \in {\cal N},i \ne j} {{\omega ^i}{\omega ^j}{D_{\rm KL}}\left( {\left. {\pi _M^i} \right\|\pi _M^j} \right)} } ,  \label{eq:KLDMO} \\
	{D_{\rm KL}}\left( { {\sum\limits_{i \in {\cal N}} {{\omega ^i}\pi _\beta ^i} } \|{{\overline \pi }_\beta }} \right) & \le \sum\limits_{i \in {\cal N}} {\sum\limits_{j \in {\cal N},i \ne j} {{\omega ^i}{\omega ^j}{D_{\rm KL}}\left( {\left. {\pi _\beta ^i} \right\|\pi _\beta ^j} \right)} } . \label{eq:KLDLO}
	\end{align}
\end{theorem}
The proof of Theorem \ref{the:T2} is given in Appendix \ref{app:the2}.

\section{Dealing with different fields-of-view}
The previous section has proposed to fuse LRFS densities
adopting the MIL rule.
%which can deal with the low probability of detection situation.
Such a rule has been developed under the pre-condition that
all the involved LRFS densities represent the multi-object LRFS
in the same FoV.
However, this is not always the case
due to the fact that, in practice,
the detection zone of each sensor is limited.
In order to cover a large-scale \textit{area of interest} (AoI),
many sensors with limited FOVs are deployed within the AoI.
In this section,
MIL fusion is extended to handle the problem of 
multi-object density fusion with different FoVs.

\subsection{On difficulties of MIL fusion with different FoVs}   \label{sec:diff}

Recall that any LRFS density ${\boldsymbol \pi}=(p,f)$ is completely characterized by 
its JEP $p$
and CJPDF $f$.
%For the sake of convenience, we abbreviate (\ref{eq:GLRFS}) as
%${\boldsymbol \pi } = \left( {p,f} \right)$.
Let us consider the problem of fusing LRFS densities ${\boldsymbol \pi }^i=(p^i,f^i), i\in{\cal N}$,
in different FoVs with their respective local label space $\mathbb L^i$,
where $\mathbb L^i$ may be (partially) overlapped or totally disjoint with $\mathbb L^j$,
%, i.e. $\mathbb L^i \cap \mathbb L^j \ne \emptyset$, 
for $i,j \in {\cal N}, i \ne j$.
Notice that it is assumed here that the labels among local densities
have been perfectly matched.
The purpose is to find the global LRFS density
${\boldsymbol \pi }=(p,f)$
defined on the label space $\mathbb L = \mathop \cup\nolimits_{i \in {\cal N}} {\mathbb L^i}$
that leads to MIL.
As indicated in Proposition \ref{pro:P0},
the fused LRFS density $\overline {\boldsymbol \pi}=(\overline f,\overline p)$ computed by the MIL rule 
amounts to fusing the JEPs and CJPDFs separately,
and the resulting fused JEP $\overline p$ of any label set $L \subseteq {\mathbb L}$
and its corresponding CJPDF $\overline f$
are equal to the weighted sums of the involved JEPs 
and CJPDFs defined on the same label set $L$.
However, if the MIL rule is directly adopted
without additional care
to fuse LRFS densities with different FoVs,
the resulting fused density might not correctly reflect
the joint existence of all targets that are located
in both the common and exclusive FoVs of the agents.
%Before illustrating the reasons,
%an example is given in advance to strengthen such fact.
The reason leading to such difficulties is that,
for general LRFS densities,
the labels are not independent of each other.
%\textcolor{red}{Instead, they are jointly depicted by the basic element ``label set",
%inside which all the labels share the same existence probability.}
In the case in which each local LRFS density ${\boldsymbol \pi }^i, i \in {\cal N}$, carries only information 
within its own FoV,
%the local JEPs 
it turns out that $p^i(L)=0$, if $L \cap \left( {\mathbb L\backslash {\mathbb L_i}} \right)  \ne \emptyset$.
As a result,
if the existing targets are located inside the exclusive FoVs of sensor nodes,
they cannot be detected jointly.
In order to better illustrate this point,
an example is given hereafter.
\begin{figure}[tb]
	\centering {
		\begin{tabular}{ccc}
			\includegraphics[width=0.35\textwidth]{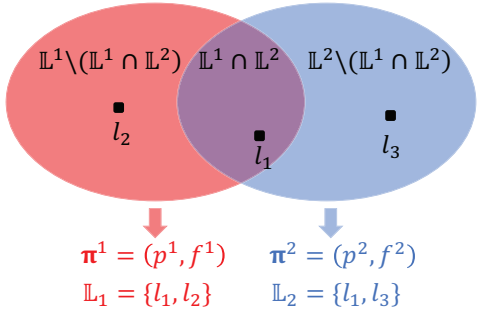}\\
		\end{tabular}
	}
	\caption{Fusion of two LRFS densities in two partially overlapped FoVs.}
	\vspace{-0.5\baselineskip}
	\label{Fig:TDDFoV}
\end{figure}

\begin{example}
	Consider the problem of fusing two LRFS densities ${\boldsymbol \pi }^1$ and ${\boldsymbol \pi }^2$ 
	in two partially overlapped FoVs,
	which are defined on label spaces $\mathbb L^1$ and $\mathbb L^2$ respectively,
	as shown in Fig. \ref{Fig:TDDFoV}.
	Suppose that the fusion weights of the two local LRFS densities are both $0.5$,
	their respective JEPs are given by
	${p^1}\left( \emptyset  \right) = 0.1,\, {p^1}\left( {\left\{ {{l_1}} \right\}} \right) = 0.05,\, {p^1}\left( {\left\{ {{l_2}} \right\}} \right) = 0.05,\, {p^1}\left( {\left\{ {{l_1},{l_2}} \right\}} \right) = 0.8$
	and ${p^2}\left( \emptyset  \right) = 0.05,\, {p^2}\left( {\left\{ {{l_1}} \right\}} \right) = 0.05,\, {p^2}\left( {\left\{ {{l_3}} \right\}} \right) = 0.05,\, {p^2}\left( {\left\{ {{l_1},{l_3}} \right\}} \right) = 0.85$.
	If we directly apply MIL fusion,
	the JEP of the fused LRFS density is computed as
	$ p\left( \emptyset  \right) = 0.1,\,  p\left( {\left\{ {{l_1}} \right\}} \right) = 0.05,\,  p\left( {\left\{ {{l_2}} \right\}} \right) = 0.025,\,  p\left( {\left\{ {{l_3}} \right\}} \right) = 0.025,\, p\left( {\left\{ {{l_1},{l_2}} \right\}} \right) = 0.4,\,  p\left( {\left\{ {{l_1},{l_3}} \right\}} \right) = 0.425,\, p\left( {\left\{ {{l_2},{l_3}} \right\}} \right) = 0,\, p\left( {\left\{ {{l_1},{l_2},{l_3}} \right\}} \right) = 0$.
	Even if the tracks in the exclusive FoVs are copied into the fused JEP $p$,
	the tracks $\{{l_1},{l_2},{l_3}\}$ cannot be jointly detected
	since $p(\{{l_1},{l_2},{l_3}\}) = 0$. 
\end{example}

\begin{remark}
	It should be noticed that
	if the involved LRFS densities are LMB,
	the above mentioned difficulties are no longer present.
	This is due to the fact that the LMB density directly relies on labels rather than label sets,
	and the existence probabilities and PDFs of labels are independent of each other.
	As a result,
	it can be directly checked that by utilizing the results of Proposition \ref{pro:P2},
	the EPs of labels of the fused LMB density will not go to zero,
	thus the JEP of existing labels will also not become null.
	Interestingly,
	if Proposition \ref{pro:P1} is adopted to fuse LMB densities, 
	the resulting JEP (\ref{eq:JEPP2}) becomes zero, thus the above mentioned difficulties still exist.
\end{remark}

\subsection{Fusion of independent LRFS densities based on MIL rule}
In order to overcome the difficulties raised in Section \ref{sec:diff},
in this subsection
we propose to perform fusion of local LRFS densities with different FoVs
by adopting the MIL rule on their respective sub-densities
defined on mutually disjoint label subspaces.

Suppose that the global label space $\mathbb L$ has been decomposed into $M$ disjoint subspaces,
i.e. $\mathbb L = \mathop  \uplus \nolimits_{m = 1}^M {{\mathbb L_m}}$ with ${\mathbb L_m}\bigcap{\mathbb L_{m'}} = \emptyset$ if $m \ne m'$.
For each subspace $\mathbb L_m$,
an LRFS density ${\boldsymbol \pi}_m= \left( {{p_m},{f_m}} \right)$ has been properly defined.
Accordingly, for an LRFS $\bf X$ whose elements are defined over the global label space $\mathbb L$, 
its LRFS density can be computed as
\begin{align}
{\boldsymbol \pi} \left( {\bf{X}} \right) = \prod\limits_{m = 1}^M {{{\boldsymbol \pi} _m}\left( {\bf{X}}_m \right)} ,   \label{eq:ProdIND}
\end{align}
where ${{\bf{X}}_m}$ is such that ${\cal L}\left( {{{\bf{X}}_m}} \right) = {\cal L}\left( {\bf{X}} \right)\bigcap {{\mathbb L_m}}$,
and ${\cal L}$ denotes the projection from LRFS to its counterpart label set, see \cite[Definition 1]{vo2013labeled}.
%and ${{\boldsymbol \pi} _m} = \left( {{p_m},{f_m}} \right)$ is the LRFS density defined only on label space $\mathbb L_m$.
Since ${\boldsymbol \pi} _m$ itself is an LRFS density, we have
\begin{align}
\int {{{\boldsymbol \pi} _m}\left( {{{\bf{X}}_m}} \right)\delta {{\bf{X}}_m}}  &= 1,  \label{eq:Prop1} \\
{{\boldsymbol \pi} _m}\left( {{{\bf{X}}_m}} \right) &= 0,\;{\rm if}\;{\cal L}\left( {{{\bf{X}}_m}} \right)\bigcap {\left\{ {\mathbb L\backslash {\mathbb L_m}} \right\}}  \ne \emptyset .  \label{eq:Prop2}
\end{align}
For the sake of convenience,
we introduce the shorthand notation ${\boldsymbol \pi}  = \left\{ {\boldsymbol \pi}_m \right\}_{m = 1}^M$.
%Therefore,
%for each agent $i\in {\cal N}$,

Unfortunately, providing 
all the local sub-densities ${\boldsymbol \pi}^i  = \{ {\boldsymbol \pi}_m^i \}_{m = 1}^M$, for $i \in {\cal N}$,
if the MIL rule is directly applied,
the resulting global density 
\begin{align}
\overline {\boldsymbol \pi} \left( {\bf{X}} \right) = \sum\limits_{i \in {\cal N}} \omega^i {\prod\limits_{m = 1}^M {{\boldsymbol \pi} _m^i\left( {{{\bf{X}}_m}} \right)} } 
\end{align} 
would lose independence among label subspaces,
thus providing the difficulties mentioned in Section \ref{sec:diff}.
In this section, similar to finding the ``best" global density 
that belongs to the same family of local ones and leads to MIL,
we propose to find
the ``best" global LRFS density $\overline {\boldsymbol \pi}$ that is independently defined on the label subspaces $\mathbb L_1, \ldots,\mathbb L_M$,
i.e. $\overline {\boldsymbol \pi}  = \left\{ \overline{\boldsymbol \pi}_m \right\}_{m = 1}^M$,
and leads to MIL.
Accordingly, the MIL rule can be properly re-defined as
\begin{align}
\overline{\boldsymbol \pi}  = \arg \mathop {\min }\limits_{\left\{ {{{\boldsymbol \pi} _m}} \right\}_{m = 1}^M} \sum\limits_{i \in {\cal N}} {{\omega ^i}\cdot {D_{\rm KL}}\left( {\left. {\prod\limits_{m = 1}^M {{\boldsymbol \pi} _m^i} } \right\|\prod\limits_{m = 1}^M {{{\boldsymbol \pi} _m}} } \right)} .   \label{eq:MWILInd}
\end{align}
The solution to the revised MIL fusion rule (\ref{eq:MWILInd}) can be found according to the following proposition.
\begin{proposition}\label{Pro:FoID}
	Given local LRFS densities ${\boldsymbol \pi}^i  = \{ {\boldsymbol \pi}_m^i \}_{m = 1}^M$, for $i \in {\cal N}$,
	the ``best" global LRFS density $\overline {\boldsymbol \pi}  = \left\{ \overline{\boldsymbol \pi}_m \right\}_{m = 1}^M$ that 
	is independently defined on $M$ label subspaces, $\mathbb L_1,\ldots,\mathbb L_M$
	and leads to MIL is given by
	\begin{align}
	{\overline {\boldsymbol \pi} _m}\left( {\bf{X}} \right) = \sum\limits_{i \in {\cal N}} {{\omega ^i} \cdot {\boldsymbol \pi} _m^i\left( {\bf{X}} \right)} , \; m = 1,\ldots,M.
	\end{align}
\end{proposition}
Proof: see Appendix \ref{app:P3}.

\begin{remark}
	It should be noticed that
	if the GCI fusion rule is adopted to fuse local densities that are
	independently defined on label subspaces,
	the resulting global density turns out to be independently defined on the same label subspaces.
	To see this,
	let us compute the global density following the GCI rule as follows
	\begin{align}
	{\boldsymbol \pi} \left( {\bf{X}} \right) &= \frac{{\prod\limits_{i \in {\cal N}} {\prod\limits_{m = 1}^M {{{\left[ {{\boldsymbol \pi} _m^i\left( {{{\bf{X}}_m}} \right)} \right]}^{{\omega ^i}}}} } }}{{\int { \cdots \int {\prod\limits_{i \in {\cal N}} {\prod\limits_{m = 1}^M {{{\left[ {{\boldsymbol \pi} _m^i\left( {{{\bf{X}}_m}} \right)} \right]}^{{\omega ^i}}}\delta \left\{ \bigcup\limits_{m = 1}^M {{{\bf{X}}_m}} \right\} } } } } }}   \nonumber  \\
	&= \prod\limits_{m = 1}^M {\frac{{\prod\limits_{i \in {\cal N}} {{{\left[ {{\boldsymbol \pi} _m^i\left( {{{\bf{X}}_m}} \right)} \right]}^{{\omega ^i}}}} }}{{\int {\prod\limits_{i \in {\cal N}} {{{\left[ {{\boldsymbol \pi} _m^i\left( {{{\bf{X}}_m}} \right)} \right]}^{{\omega ^i}}}} \delta {{\bf{X}}_m}} }}} .
	\end{align}
	Defining
	\begin{align}
	{{\boldsymbol \pi} _m}\left( {{{\bf{X}}_m}} \right) = \frac{{\prod\limits_{i \in {\cal N}} {{{\left[ {{\boldsymbol \pi} _m^i\left( {{{\bf{X}}_m}} \right)} \right]}^{{\omega ^i}}}} }}{{\int {\prod\limits_{i \in {\cal N}} {{{\left[ {{\boldsymbol \pi} _m^i\left( {{{\bf{X}}_m}} \right)} \right]}^{{\omega ^i}}}} \delta {{\bf{X}}_m}} }},
	\end{align}
	the above conclusion can be immediately drawn.
\end{remark}

\subsection{Decomposition of LRFS densities}
Previous sections have shown that
if the global label space $\mathbb L$ is made up of $M$ mutually disjoint label subspaces
and local sub-densities for the corresponding label subspaces have been properly defined,
the fused density can be found by performing fusion with respect to the sub-densities
on each label subspace.
However, in practice the local density ${\boldsymbol \pi}^i$ at each agent $i \in {\cal N}$
is defined within its whole FoV,
thus is not equal to the product of sub-densities defined on the label subspaces.
In this subsection,
we seek for a method to factorize an LRFS density ${\boldsymbol \pi}$ 
into $M$ mutually independent sub-densities defined on label subspaces
by minimizing  the KLD from the re-constructed density to the original one, as shown in the following Proposition.

\begin{proposition} \label{pro:dec}
	Suppose that a general LRFS density ${\boldsymbol \pi}=(p,f)$ is defined on the label space $\mathbb L$.
	Then, the best decomposition of ${\boldsymbol \pi}$ into $M$ sub-densities $\left\{ {\boldsymbol \pi}_m \right\}_{m = 1}^M$ 
	defined on $M$ mutually disjoint label spaces $\mathbb L_1, \ldots,\mathbb L_M$ 
	minimizing the KLD from the re-constructed density (\ref{eq:ProdIND}) to the original one
	can be found as
	${{\boldsymbol \pi} _m} = \left( {{p_m},{f_m}} \right)$ given by
	\begin{align}
	{p_m}\left( {{L_m}} \right) &= \sum\limits_{L:L \supseteq L_m} {p\left( L \right)} ,  \label{eq:decwei} \\
	{f_m}\left( {\left. {{X_m}} \right|{L_m}} \right) &= \arg \mathop {\min }\limits_{f'} \sum\limits_{L:L \supseteq L_m} {\tilde \omega \left( L \right){D_{\rm KL}}\left( {\left. {{{\tilde f}_m}} \right\|f'} \right)} ,  \label{eq:decspdf}
	\end{align}
	where
	\begin{align}
	\bar \omega \left( L \right) &= \frac{{p\left( L \right)}}{\sum\limits_{L':L' \supseteq L_m} {p\left( {L'} \right)} }, \\
	{{\tilde f}_m}\left( {\left. {{X_m}} \right|L} \right) &= \int {f\left( {\left. X \right|L} \right)d \left( {X\backslash {X_m}} \right)} .
	\end{align}
\end{proposition}
Proof: see Appendix \ref{app:P5}. \qed 

A  similar splitting of (\ref{eq:decwei}) in Proposition \ref{pro:dec} can be found for the $\delta$-GLMB density
in \cite{beard2018solution} where the aim is to deal with large-scale multitarget tracking with a single sensor.
Here Proposition \ref{pro:dec} provides the following extensions with respect to \cite{beard2018solution}:
\begin{itemize}
	\item decomposition of an arbitrary LRFS density; % can be performed; %with the results of Proposition \ref{pro:dec};
	\item more importantly,
	by means of  (\ref{eq:decspdf}) in
	Proposition \ref{pro:dec},
	the CJPDFs of the decomposed LRFS densities are also provided,
	while only computation of the JEPs is addressed in \cite{beard2018solution}.
\end{itemize}

Please notice that the CJPDFs $f_m$ of the sub-densities $\boldsymbol{\pi}_m$ are not given explicitly by (\ref{eq:decspdf}) 
but as the result of the minimization of the MIL criterion.
%Note that such optimization problem is consistent with the MIL rule given in Proposition \ref{pro:P0}.
Thus,
Proposition \ref{pro:dec} can be easily extended to any specific class of  LRFS densities.
For instance,
for M$\delta$-GLMB densities whose CJPDF is independent among tracks,
$f_m$ is computed by directly applying (\ref{eq:Ff}).
Furthermore,
according to (\ref{eq:decwei}),
if ${\boldsymbol \pi}$ is decomposed to a label space $\mathbb L_{m'}$ such that $\mathbb L_{m'} \cap \mathbb L = \emptyset$,
for instance ${\boldsymbol \pi}^1$ is decomposed to $\mathbb L^2 \backslash( \mathbb L^1  \cap \mathbb L^2)$ in the example of Fig. \ref{Fig:EXA},
the resulting sub-density ${\boldsymbol \pi}_{m'}$ will always be null given any LRFS, i.e. 
${\boldsymbol \pi}_{m'}({\bf X}) = 0$ for ${\cal L}({\bf X}) \subseteq {\mathbb L}_{m'}$.

\begin{remark}
	Due to the fact that the Bernoulli components (BCs) of an LMB density are mutually independent,
	i.e. the LMB density is by construction decomposed into $|\mathbb L|$ subspaces 
	where each subspace has only one label,
	the MIL fusion rule can directly be adopted to fuse LMB densities defined in different FoVs.
\end{remark}

\section{Implementation issues}
\subsection{Fusion of CJPDFs}
It has been shown in Propositions \ref{pro:P0} that
MIL fusion of LRFS densities 
amounts to separately fusing the JEPs and CJPDFs.
Since the JEP is a discrete density,
fusion of JEPs is quite straightforward.
%Whereas, in practice,
%the CJPDF of LRFS density is often assumed to be approximately represented by Gaussian mixtures (GMs) 
%or particle sets \cite{vo2014labeled},
%thus additional cares should be taken to promote the overall performance.
In this subsection, implementation issues relative to MIL fusion of CJPDFs are discussed.
Since
fusion of M$\delta$-GLMB and LMB densities is of particular interest in practice,
and MIL fusion of CJPDFs of these two densities
is carried out independently of labels (see Propositions \ref{pro:P1} and \ref{pro:P2}),
we focus on the implementation of MIL fusion on a single label $l$ with local PDFs given as $f_l^i$, for $i \in {\cal N}$.
Notice that, in practice,
the PDF of a label is often assumed to be approximately represented by a \textit{Gaussian mixture} (GM) 
or a \textit{particle set} (PS) \cite{vo2014labeled}.
In the rest of this subsection,
the implementation issues relative to these two representations are separately discussed.

\textbf{Fusion with GMs:} 
Suppose now that  the PDF $f^i_l$ is approximated by a GM as
\begin{align}
f_l^i\left( x \right) \cong \sum\limits_{m = 1}^{J^i_l} {\alpha ^{i,m}{\cal G}\left( {x;\mu ^{i,m},P^{i,m}} \right)} ,
\end{align}
where ${\cal G}\left( {x;\mu ,P} \right)$ denotes a Gaussian PDF with mean $\mu$ and covariance matrix $P$.
Then, the PDF of the fused RFS density is given by
\begin{align}
{f_l}\left( x \right) = \sum\limits_{i \in {\cal N}} {\sum\limits_{m = 1}^{J_l^i} {{{\tilde \omega }^i}\alpha ^{i,m}{\cal G}\left( {x;\mu ^{i,m},P^{i,m}} \right)} } .
\end{align}
where ${\tilde \omega }^i$ is computed via (\ref{eq:OmegaMGLMB}) if local LRFS densities are M$\delta$-GLMB
or (\ref{eq:OmegaLMB}) if local LRFS densities are LMB.
Note that the number of \textit{Gaussian components} (GCs) increases to ${\sum\nolimits_{i \in {\cal N}} {J^i_l} }$ after fusion, 
which leads to an increase of computational burden.
Hence, suitable pruning and merging procedures \cite[Table II]{vo2006gaussian} should be performed in order to reduce the number of GCs.

\textbf{Fusion with PSs:}
Suppose that the PDF $f_l^i$ is approximated by a set of particles as 
\begin{align}
f_l^i\left( x \right) \cong \sum\limits_{m = 1}^{J_l^i} {\alpha^{i,m}{\delta _{x^{i,m}}}\left( x \right)} ,
\end{align}
where ${\delta _x}(\cdot)$ is the Dirac delta centered at $x$..
Then, the fused PDF is given by
\begin{align}
{f_l}\left( x \right) = \sum\limits_{i \in {\cal N}} {{{\tilde \omega }^i}f_l^i\left( x \right)}  = \sum\limits_{i \in {\cal N}} {\sum\limits_{m = 1}^{J_l^i} {{{\tilde \omega }^i}\alpha ^{i,m}{\delta _{x^{i,m}}}\left( x \right)} } ,  \label{eq:SMCfus}
\end{align}
Similarly to GM implementation, the number of particles increases to ${\sum\nolimits_{i \in {\cal N}} {J_l^i} }$ after fusion via (\ref{eq:SMCfus}),
thus leading to an increase of computational load at the next time instance.
Then, a resampling step \cite[Section III-F]{vo2005sequential} should be performed to select a total amount of $J_l$ (which can be determined by the corresponding JEP of the label set) particles.

\begin{remark}
	When performing GCI fusion with GM implementation, 
	the need arises to approximately compute the power of GMs.
	Although there exist approximate methods \cite{gunay2016chernoff} to accomplish such a task with satisfactory accuracy,
	a non negligible extra computational load  is required to perform such approximation.
	By contrast, MIL fusion of GMs directly provides a fused GM without any approximation,
	thus providing enhanced accuracy and computational savings.
\end{remark}

\begin{remark}
	Normally, a huge number of particles is required to reasonably approximate the PDF,
	thus implying heavy transmission load.
	In order to reduce communication bandwidth within the WSN,
	one can further approximate particle sets by GMs with reduced number of GCs \cite{li2018local}.
	In this way, fusion can be performed via GM implementation on the approximated GMs.
	After fusion, the resulting GM can be converted back to SMC representation by mean of a suitable sampling method \cite{li2018local}.
\end{remark}

\subsection{Solving the label mismatching problem}  \label{sec:LM}
The MIL fusion of LRFS densities proposed in Section \ref{sec:MWIL} is based on the assumption that
all the involved local LRFS densities are defined on the same label space.
As a matter of fact,
such assumption is impractical in many applications,
for instance:
\begin{itemize}
	\item[-] when the tracks are initialized by the adaptive birth model \cite{ristic2012adaptive} at each agent (with different number of measurements at each time),
	the numbers of birth BCs at each time are different,
	thus it is not possible to ensure to assign the same track with the same label;
	\item[-] even though tracks are initialized with the same prior information at each agent,
%	since the pruning strategy \cite{vo2014labeled,reuter2014labeled} are normally adopted at local filter to control the number of computational load,
	because of target miss-detections  and false alarms,
	it is also difficult to ensure matching of the label sets of all agents.
\end{itemize}
Hence, the practical implementation of MIL fusion of LRFS densities must be able to solve also the label mismatching problem.
%In this section, we aim to solve the label mismatching problem among LRFS densities.
It has been shown in \cite{li2019computationally} that,
for a non-LMB density, it is convenient to find the ``best" LMB approximation \cite[Algorithm 1]{garcia2016track}
and then perform label matching among LMB densities.

Let us therefore consider the problem of label matching between two LMB densities 
${\boldsymbol \pi} _\beta ^1 = {\left\{ {\left( {r_l^1,f_l^1} \right)} \right\}_{l \in {{\mathbb L}^1}}}$ and 
${\boldsymbol \pi} _\beta ^2 = {\left\{ {\left( {r_l^2,f_l^2} \right)} \right\}_{l \in {{\mathbb L}^2}}}$.
Without loss of generality, 
it is assumed that $\left| {{{\mathbb L}^1}} \right| \ge \left| {{{\mathbb L}^2}} \right|$.
It has been shown in \cite{li2019computationally} how
associating the track labels of two LMB densities can be achieved by solving  a \textit{ranked assignment problem} (RAP).
To this end, a cost (square) matrix ${\cal C}$ with dimension $\left| {{{\mathbb L}^1}} \right|$ 
(i.e. the larger label space cardinality)
is constructed,
in which the value of each element ${{c_{n_1,n_2}}}$ for $n_1=1,\cdots,\left| {{{\mathbb L}^1}} \right|$ and
$n_2=1,\ldots,\left| {{{\mathbb L}^2}} \right|$
is defined as the so-called GCI divergence ${D_{\rm GCI}}\left( {\left. {{l_{{n_1}}}} \right\|{l_{{n_2}}}} \right)$ (i.e. the cost when performing label-wise GCI fusion between the BC with label $l_{n_1}$ in ${\boldsymbol \pi} _\beta ^1$
and the BC with label $l_{n_2}$ in ${\boldsymbol \pi} _\beta ^2$, see \cite[Appendix]{battistelli2013consensus}) 
given by
\begin{align}
& {D_{\rm GCI}}\left(l_{n_1},l_{n_2} \right) \nonumber \\
&=  - \log \left[ {{{\left( {1 - r_{{l_{{n_1}}}}^1} \right)}^{{\omega ^1}}}{{\left( {1 - r_{{l_{{n_2}}}}^2} \right)}^{{\omega ^2}}} + {{\left( {r_{{l_{{n_1}}}}^1} \right)}^{{\omega ^1}}}{{\left( {r_{{l_{{n_2}}}}^2} \right)}^{{\omega ^2}}}} \right. \nonumber  \\
& \quad  \times \left. {\int {{{\left[ {f_{{l_{{n_1}}}}^1\left( x \right)} \right]}^{{\omega ^1}}}{{\left[ {f_{{l_{{n_2}}}}^2\left( x \right)} \right]}^{{\omega ^2}}}dx} } \right].   \label{eq:GCIdiv}
\end{align}
Note that the label set of ${\boldsymbol \pi} _\beta ^2$ is compensated by $\left| {{{\mathbb L}^1}} \right| - \left| {{{\mathbb L}^2}} \right|$ 
virtual tracks with EPs equal to zero.
With such definition,
the tracks between two label sets are matched
by finding the best assignment based on the cost matrix $\cal C$,
and such optimization problem can be solved within polynomial time adopting the Hungarian algorithm \cite{kuhn1955hungarian}.
This idea implies that every BC in the LMB density with smaller label space cardinality  (i.e. ${\boldsymbol \pi} _\beta ^2$)
will definitely be associated with a BC in the other one (i.e. ${\boldsymbol \pi} _\beta ^1$).
This method works well whenever all agents have the same FoV
and high probability of detection 
(i.e., additional BCs in ${\boldsymbol \pi} _\beta ^1$ have a high probability to be originated from clutter).
However, it has the following limitations:
\begin{itemize}
	\item[-] it cannot be adopted to handle the situation where agent FoVs are different
	since, in such a case, BCs inside the exclusive FoV of ${\boldsymbol \pi} _\beta ^2$ should not be associated to any BC in ${\boldsymbol \pi} _\beta ^1$;
	\item[-] the GCI divergence is strongly affected by the EPs of BCs, 
	as shown in Example \ref{exa:LM}.
\end{itemize}

\begin{figure}[tb]
	\centering {
		\begin{tabular}{ccc}
			\includegraphics[width=0.3\textwidth]{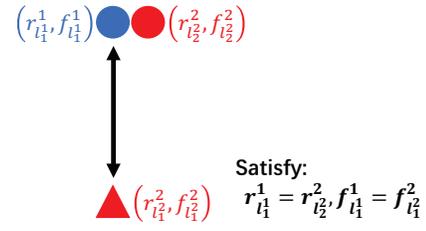}\\
		\end{tabular}
	}
	\caption{Example of two LMB densities}
	\vspace{-0.5\baselineskip}
	\label{Fig:EXA}
\end{figure}

\begin{example}  \label{exa:LM}
	Suppose that ${\boldsymbol \pi} _\beta ^1$ (with fusion weight $\omega$) consists of a single BC and
	${\boldsymbol \pi} _\beta ^2$ (with fusion weight $1-\omega$) consists of two BCs, where
	$f_{{l_1^1}}^1 = f_{{l_2^2}}^2,r_{{l_1^1}}^1 = r_{{l_1^2}}^2$
	as shown in Fig. \ref{Fig:EXA}.
	Note that this situation could happen when
	both tracks $l_1^1$ in ${\boldsymbol \pi} _\beta ^1$ 
	and $l_1^2$ in ${\boldsymbol \pi} _\beta ^2$ 
	are miss-detected.
	In practice, it is desired to match ${l_1^1}$ with ${l_2^2}$,
	i.e. ${D_{\rm GCI}}\left( {l_1^1,l_2^2} \right) < {D_{\rm GCI}}\left( {l_1^1,l_1^2} \right)$,
	due to the fact that $l_1^1$ and $l_2^2$ are located at the same position.
	However, mismatching happens when
	\begin{align}
	{C_\omega }\left( {\beta \left( {r_{l_1^1}^1} \right),\beta \left( {r_{l_2^2}^2} \right)} \right) \le 1 - r_{l_1^1}^1 + r_{l_1^1}^1 \cdot {C_\omega }\left( {f_{l_1^1}^1,f_{l_1^2}^2} \right),  \label{eq:MIS}
	\end{align}
	where ${\beta \left( r \right)}$ represents a Bernoulli distribution with probability $r$
	and $C_\omega$ denotes the Chernoff $\omega$-coefficient defined as \cite{nielsen2011chernoff}
	\begin{align}
	{C_\omega }\left( {{f^1},{f^2}} \right) = \int {{{\left[ {{f^1}\left( x \right)} \right]}^\omega }{{\left[ {{f^2}\left( x \right)} \right]}^{1 - \omega }}dx},
	\end{align}
	with the integral replaced by summation when $f^1$ and $f^2$ are defined over a discrete space (e.g. Bernoulli distribution).
	The proof of (\ref{eq:MIS}) is omitted since it can be directly obtained substituting the parameters of BCs into the corresponding definitions.
	Due to the fact that $0 \le C_\omega \le 1$,
	and $C_\omega \left( {{f^1},{f^2}} \right)$ tends to  $1$ when $f^1$ and $f^2$ are similar,
	it can be seen immediately that when $r_{l_1^1}^1$ is extremely low, the right-hand-side of (\ref{eq:MIS})
	will be close to $1$, which means that 
	mismatching might happen when there exist miss-detections among agents.
\end{example}

Therefore, in this subsection,
we propose to solve the label mismatching problem 
by constructing a modified RAP.
%also by a different rank assignment strategy.
Specifically, tye following cost matrix ${\cal C}$ with dimension $\left( {\left| {{{\mathbb L}^1}} \right| + 1} \right) \times \left( {\left| {{{\mathbb L}^2}} \right| + 1} \right)$ is defined:
\begin{align}
{\cal C} = \left[ {\begin{array}{*{20}{c}}
	{{c_{1,1}}}& \cdots &{{c_{1,\left| {{{\mathbb L}^2}} \right| + 1}}}\\
	\vdots & \ddots & \vdots \\
	{{c_{\left| {{{\mathbb L}^1}} \right| + 1,1}}}& \ldots &{{c_{\left| {{{\mathbb L}^1}} \right| + 1,\left| {{{\mathbb L}^2}} \right| + 1}}}
	\end{array}} \right],
\end{align}
in which the entry $c_{n_1,n_2}$ represents the cost of assigning
the BC $( {r_{{l_{{n_1}}}}^1,f_{{l_{{n_1}}}}^1})$ of ${\boldsymbol \pi} _\beta ^1$
to the BC $( {r_{{l_{{n_2}}}}^2,f_{{l_{{n_2}}}}^2})$ of ${\boldsymbol \pi} _\beta ^2$.
Further $c_{n_1,{\left| {{{\mathbb L}^2}} \right| + 1}}$ denotes 
the cost of regarding $( {r_{{l_{{n_1}}}}^1,f_{{l_{{n_1}}}}^1})$ of ${\boldsymbol \pi} _\beta ^1$ as unassociated while
$c_{{\left| {{{\mathbb L}^1}} \right| + 1},n_2}$ denotes the cost of regarding $( {r_{{l_{{n_2}}}}^1,f_{{l_{{n_2}}}}^2})$
of ${\boldsymbol \pi} _\beta ^2$ as unassociated.
Finally, we artifically set ${c_{\left| {{{\mathbb L}^1}} \right| + 1,\left| {{{\mathbb L}^2}} \right| + 1}} = \infty$.

Motivated by the above mentioned limitations of GCI divergence,
we define the entry $c_{n_1,n_2}$ as the divergence 
that considers only the PDF of the BCs, i.e.
\begin{align}
{c_{{n_1},{n_2}}} \hspace{-0.1cm} \buildrel \Delta \over = \hspace{-0.1cm} \left\{ \begin{array}{l} \hspace{-0.2cm}
D\left( {f_{{l_{{n_1}}}}^1,f_{{l_{{n_2}}}}^2} \right),\; 1 \le {n_1} \le \left| {{\mathbb L^1}} \right| \,\& \, 1 \le {n_2} \le \left| {{\mathbb L^2}} \right|\\
\hspace{-0.2cm} \infty , \; \quad \quad \quad \quad \quad \,{n_1} = \left| {{\mathbb L^1}} \right| + 1 \,\& \, {n_2} = \left| {{\mathbb L^2}} \right| + 1\\
\hspace{-0.2cm} {T_D},\; \quad \quad \quad \quad \quad {\rm otherwise}
\end{array} \right. \hspace{-0.25cm} ,
\end{align}
where $T_D$ is the matching threshold that represents the largest PDF divergence
that the same target could have among agents,
and $D(\cdot)$  represents an information-theoretic discrepancy among PDFs.
There are several candidates that can be adopted to this end,
such as:
\begin{itemize}
	\item[-] Jensen-Shannon divergence $D_{\rm JS}$, which is also known as the symmetric KLD, and is defined as
	\begin{align}
	{D_{\rm JS}}\left( {{f^1},{f^2}} \right) = \frac{1}{2}\left[ {{D_{\rm KL}}\left( {\left. {{f^1}} \right\|{f^2}} \right) + {D_{\rm KL}}\left( {\left. {{f^2}} \right\|{f^1}} \right)} \right];  \label{eq:DJS}
	\end{align}
	\item[-] Cauchy-Schwarz divergence $D_{\rm CS}$, which is defined as
	\begin{align}
	{D_{\rm CS}}\left( {{f^1},{f^2}} \right) =  - \log \left\{ {\frac{{\int {{f^1}(x){f^2}(x)dx} }}{{\sqrt {\int {{{\left[ {{f^1}(x)} \right]}^2}dx}  \cdot \int {{{\left[ {{f^2}(x)} \right]}^2}dx} } }}} \right\} .   \label{eq:DCS}
	\end{align}
\end{itemize}
\begin{remark}
	Concerning the computation of  information-theoretic discrepancies, 
	the following facts needs to be clarified.
	\begin{enumerate}[1)]
		\item When the PDFs of local LMB densities are approximately represented with GMs,
		the CSD between PDFs can be computed analytically  while, on the other hand,
		the computation of the KLD doe not admit an analyitical form.
		In the latter case, an approximate solution can be obtained with the aid of a sigma-point representation of the GMs; 
		the details can be found in \cite[Appendix A]{gao2018event};
		\item When the PDFs of local LMB densities are approximately represented with particle sets,
		both KLD and CSD cannot be accurately computed unless
		a sufficient amount of particles among the involved PDFs are overlapped.
		Therefore, in this case,
		it is suggested to further approximate the particle sets by GMs \cite{li2018local} 
		and then adopt the method discussed in 1).
	\end{enumerate}
\end{remark}

In order to better illustrate the proposed strategy, 
it is useful to define the assignment matrix $\cal S$ as
\begin{align}
{\cal S} = \left[ {\begin{array}{*{20}{c}}
	{{s_{1,1}}}& \cdots &{{s_{1,\left| {{{\mathbb L}^2}} \right| + 1}}}\\
	\vdots & \ddots & \vdots \\
	{{s_{\left| {{{\mathbb L}^1}} \right| + 1,1}}}& \ldots &{{s_{\left| {{{\mathbb L}^1}} \right| + 1,\left| {{{\mathbb L}^2}} \right| + 1}}}
	\end{array}} \right],
\end{align}
where $s_{n_1,n_2} = 1$ 
if BC $( {r_{{l_{{n_1}}}}^1,f_{{l_{{n_1}}}}^1})$
is assigned to $( {r_{{l_{{n_2}}}}^2,f_{{l_{{n_2}}}}^2})$
and otherwise $s_{n_1,n_2} = 0$.
Note that, $s_{n_1,{\left| {{{\mathbb L}^2}} \right| + 1}}=1$ means 
$( {r_{{l_{{n_1}}}}^1,f_{{l_{{n_1}}}}^1})$ remains unassigned and similarly
$s_{{\left| {{{\mathbb L}^1}} \right| + 1},n_2}  = 1$ that $( {r_{{l_{{n_2}}}}^1,f_{{l_{{n_2}}}}^2})$
is unassigned;
moreover, ${s_{\left| {{{\mathbb L}^1}} \right| + 1,\left| {{{\mathbb L}^2}} \right| + 1}} \equiv 0$.
Then, the problem turns out to find the best assignment ${\cal S}^*$ that minimizes the global cost, i.e.
\begin{align}
{{\cal S}^*} \buildrel \Delta \over = \arg \mathop {\min }\limits_{\cal S} \sum\limits_{{n_1} = 1}^{\left| {{\mathbb L^1}} \right| + 1} {\sum\limits_{{n_2} = 1}^{\left| {{\mathbb L^2}} \right| + 1} {{s_{{n_1},{n_2}}} \cdot {c_{{n_1},{n_2}}}} }  = \arg \mathop {\min }\limits_{\cal S} {\rm tr}\left( {{{\cal S}^ \top }{\cal C}} \right), \label{eq:RAP}
\end{align}
where ${\rm tr}(\cdot)$ denotes the trace of a matrix.
Such a linear assignment problem can be efficiently solved in polynomial time 
by the Hungarian algorithm 
%\cite{kuhn1955hungarian}.
%The Hungarian algorithm has been widely applied together with the multiple hypothesis tracker \cite{blair2000multitarget}
%and the GLMB-series filters \cite{vo2014labeled} to track multiple targets.
%Since it is a quite mature technique,
%here we will not go in details for this algorithm.
%The details can be found in
\cite{blair2000multitarget,vo2014labeled}. 

\subsection{Application of MIL fusion in the context of DMT}

One of the most important applications 
of  multi-object fusion is \textit{distributed multitarget tracking} (DMT).
In this subsection,
details of applying MIL fusion to DMT are provided.
%Recall the agent set of a multi-agent system (MAS) as $\cal N$,
%which consists of $|{\cal N}|$ agents.
The considered LRFS approach to DMT considered in this paper
consists of the following two steps recursively performed at each time $t$:
\begin{enumerate}[1)]
	\item {\bf Local filtering}. Each agent $i \in {\cal N}$,
	provided with prior ${\boldsymbol \pi}_{t-1}$
	and measurements obtained through an imperfect extraction process,
	(i.e. featuring target miss-detections and false alarms) runs 
	a multitarget tracker \cite{vo2013labeled,fantacci2015marginalized,reuter2014labeled,vo2014labeled} in order to get the local posterior ${\boldsymbol \pi}^i_{t|t}$.
	\item {\bf Information aggregation}. Based on step 1), 
	local posteriors of all agents are collected at the fusion center (or shared by a broadcast protocol like consensus \cite{xiao2007distributed})
	and then the multi-object density fusion algorithm is employed to fuse local posteriors ${\boldsymbol \pi}^i_{t|t},\; i \in {\cal N}$,  into the global density ${\boldsymbol \pi}_t$,
	and then ${\boldsymbol \pi}_t$ is utilized as  prior information for the local filtering of next iteration at each node $i\in{\cal N}$.
\end{enumerate}

%\textcolor{blue}{Since local filtering is out of the scope of this paper,
%it is assumed that
%the local posterior ${\boldsymbol \pi}_{t|t}^i$ of each agent $i \in {\cal N}$ is 
%known at each time $t$.}

In the context of DMT,
if all agents have the same FoV,
fusion can be performed directly with the proposed MIL rule,
otherwise local LRFS densities will have to be decomposed into mutually independent sub-densities
defined on suitable label subspaces
and MIL fusion is performed subspace-by-subspace.
If the local FoV %$\mathbb L^i$} 
of each agent $i \in {\cal N}$ is known,
the label subspaces can be obtained at every recursion
by looking for the closed region of the global label space.
For instance, in the example of Fig. \ref{Fig:TDDFoV},
the subspaces could be $\mathbb L_1 \buildrel \Delta \over = \mathbb L^1 \backslash ( \mathbb L^1  \cap \mathbb L^2)$,
$\mathbb L_2 \buildrel \Delta \over = \mathbb L^1  \cap \mathbb L^2$,
and $\mathbb L_3 \buildrel \Delta \over = \mathbb L^2 \backslash ( \mathbb L^1  \cap \mathbb L^2)$.
However, in practice,
it is more desirable to develop fusion rules for agents that have limited but unknown FoVs,
due to the facts that:
\begin{itemize}
	\item[-] affected by the physical conditions of  the AoI (e.g. rain, fog, etc.),
	it is hard to precisely define the FoV of each agent;
	\item[-] in some specific MAS like \textit{wireless sensor networks} (WSNs),
	the agents are powered by batteries so that as far as energy
	is consumed, the agent FoV is time-varying.
\end{itemize}

Notice that
if each agent performs well in local filtering,
the tracks within its local FoV can be correctly detected after few time recursions.
In this sense,
it is straightforward to define the label subspaces by comparing the labels
that are involved in each local LRFS densitiy
(conditioned on the fact that all the local labels have been correctly matched using the method of Section \ref{sec:LM}).
%\textcolor{blue}{
For instance again in Fig. \ref{Fig:EXA},
where ${\boldsymbol \pi}^1$ involves $l_1$ and $l_2$
while ${\boldsymbol \pi}^2$ involves $l_1$ and $l_3$,
both local densities contain track $l_1$
and $l_2, l_3$ are their respective exclusive tracks.
Then it is straightforward to define
$\mathbb L_1 = \{l_1\}$, $\mathbb L_2 = \{l_2\}$, $\mathbb L_3 = \{l_3\}$.
%}

Note that,
as far as fusion is performed,
compensated by local densities of other agents,
each agent acquires the information outside its local FoV.
As a result,
the local label space of each agent includes more and more tracks
as far as DMT is implemented.
Hence,
label subspaces should be re-defined
whenever fusion is going to be performed.
By considering all the mentioned factors,
the proposed DMT approach is outlined in Algorithm \ref{alg:DMT}.

\begin{algorithm}  
	%\SetAlgoNoLine
	\caption{DMT with LRFS (at time $t$)}
	\label{alg:DMT}
	\KwIn {${\boldsymbol \pi}_{t - 1}$} 
	Carry out local filtering (see \cite{vo2013labeled,fantacci2015marginalized,reuter2014labeled,vo2014labeled}) at each agent $i \in {\cal N}$ to compute local posteriors ${\boldsymbol \pi}_{t|t}^i$\;
	For each agent $i\in {\cal N}$, broadcast its local posterior to the fusion center\;
	Match all the involved track labels using the method illustrated in Section \ref{sec:LM}\;
	Fuse local posteriors ${\boldsymbol \pi}_{t|t}^i$ into the global density ${\boldsymbol \pi}_t$\;
	Transmit ${\boldsymbol \pi}_t$ back to each agent $i \in {\cal N}$.\\  % and set ${\boldsymbol \pi}_t^i = {\boldsymbol \pi}_t$. \\
	\KwOut{${\boldsymbol \pi}_t$}
\end{algorithm}  

\begin{remark}  \label{rem:GLMB}
	Though $\delta$-GLMB densities can be analytically fused 
	under the MIL criterion,
	the number of association hypotheses resulting in the global density 
	increases to $\sum\nolimits_{i \in {\cal N}} {\left| {{\Xi ^i}} \right|}$.
	Further, the number of association hypotheses of the $\delta$-GLMB density increases exponentially 
	during local filtering
	if no additional operation (i.e. pruning of hypotheses) is carried out.
	As a result,
	modeling the multitarget state as $\delta$-GLMB density
	for DMT requires a huge amount of memory as well as computational resources,
	thus being practically infeasible.
	In this regard, for muitarget tracking it is by far preferable to adopt M$\delta$-GLMB and LMB filters. 
\end{remark}

\begin{remark}  \label{rem:WSN}
	Note that steps $2 -4$ of Algorithm \ref{alg:DMT} are designed 
	for  MASs having a fusion center, 
	which is able to exchange information with all the agents.
	However,
	this is not always the situation since in some MASs (e.g. WSNs)
	the agents work in a peer-to-peer (P2P) manner,
	wherein each individual agent is unable to gather densities from all other agents.
	In such cases, a promising strategy is the consensus method \cite{xiao2005scheme},
	which consists of $L$ iterations of data-exchange with the neighbors and consequent fusion of the received densities with the local one to be performed at each sampling interval.
Details on the application of consensus to DMT can be found in \cite{battistelli2013consensus,fantacci2015consensus}.
\end{remark}

\section{Performance evaluation}
In this section, simulations concerning DMT over a WSN \cite{mukhopadhyay2010advances} are carried out
in order to assess the performance of MIL fusion.
Specifically,
two scenarios are considered, where the first one assumes that all the sensor nodes have the same FoV
 while the second one assumes that the sensing range of each node is limited.
Before illustrating the details of simulations,
the following statements are in order.
\begin{itemize}
	\item[-] In both scenarios,
	the MIL rule is combined with M$\delta$-GLMB and LMB densities,
	hence the local trackers of \cite{vo2014labeled} and \cite{reuter2014labeled} are respectively adopted.
	The $\delta$-GLMB density is not considered in the simulations
	since it requires a huge amount of computational and memory resources as noted in Remark \ref{rem:GLMB},
	and is therefore unsuitable for WSN applications.
	\item[-] Since the sensor nodes of a WSN are often powered by batteries,
	their computational ability,
	memory resources and communication bandwidth are limited.
	Consequently,
	all the involved multi-object densities in this section are represented by GMs.
	\item[-] As observed in Remark \ref{rem:WSN},
	the sensor nodes of a WSN work in a P2P fashion;
	hence consensus is employed in the simulations.
	In particular, we use the algorithm in \cite{fantacci2015consensus}
	but replace the ``GM-M$\delta$-GLMB Fusion" step of Table II with the results of Proposition \ref{pro:P1}
	if the multitarget state is modeled as M$\delta$-GLMB;
	or  the ``GM-LMB Fusion" step of Table II with the results of Proposition \ref{pro:P2}
	if the multitarget state is modeled as LMB.
\end{itemize}

In both scenarios,
the single target state at time $t$ is denoted as
${x_t} = [{{\xi _t}\;{{\dot \xi }_t}\;{\eta _t}\;{{\dot \eta }_t}}]^\top $,
where $[ {{\xi _t}\;{\zeta _t}} ]^\top$ and $[ {{{\dot \xi }_t}\;{{\dot \zeta }_t}} ]^\top$  are respectively position and velocity in 
Cartesian coordinates.
It is supposed that the target motion is described by the following linear white noise acceleration model
\begin{align}
x_t = A \, x_{t-1} + w_t,
\end{align}
where $w_t$ represents additive white Gaussian noise with covariance matrix $Q = {\rm diag}(16[m^2],1[m^2/s^2],16[m^2],1[m^2/s^2])$, 
and
\begin{align}
A = \left[ {\begin{array}{*{20}{c}}
	1&T&0&0\\
	0&1&0&0\\
	0&0&1&T\\
	0&0&0&1
	\end{array}} \right],
\end{align}
$T=1 [s]$ being the sampling interval.
Further, it is assumed that
each node of the WSN is able to provide both range-of-arrival (ROA) and direction-of-arrival (DOA) measurements of targets,
i.e. the measurement $z_t^i$ generated by  a target with state $x_t$, at time $t$ and in node $i\in {\cal N}$, is modeled as
\begin{align} \label{eq:MF2}
{z_t^i} = h^i\left( {{x_t}} \right) + v_t^i,
\end{align}
where $v_t^i$ is a measurement noise modeled as a zero mean Gaussian process with covariance matrix $R^i = {\rm diag}(400[m^2],\;0.64[^{o^2}])$
and 
\begin{align}
h^i\left( {{x_t}} \right) = \left[ \begin{array}{l}
\sqrt {{{\left( {{\xi _t} - {\xi ^i}} \right)}^2} + {{\left( {{\eta _t} - {\eta ^i}} \right)}^2}} \\
{\rm atan2}\left( {{\eta _t} - {\eta ^i},{\xi _t} - {\xi ^i}} \right)
\end{array} \right],
\end{align}
${\rm atan2}$ denoting the four quadrant inverse tangent.
Clutter at each sensor node has Poisson-distributed cardinality 
(expected number of targets $\lambda_c=8$ at each time) 
and uniform spatial distribution over its local FoV.

The common parameters of local tracks are set as follows:
the probability of target survival is set to $P_s=0.95$
for all sensor nodes;
the \textit{Jensen-Shannon divergence} (JSD) has been chosen as discrepancy measure
for label matching among local densities,
with matching threshold $T_D = 50$.
New-born targets are modeled as LMB,
where the number of BCs is taken equal to the number of measurements.
The EP of each BC is set to $0.01$
and the PDF is taken Gaussian,
where the position components of the mean vector are obtained by remapping measurements back to target state space
and the velocitiy components are set to zero;
the covariance matrix is set to ${\rm diag}(1600[m^2],400[m^2/s^2],1600[m^2],400[m^2/s^2])$.
The pruning and merging thresholds for GMs are set respectively to $10^{-5}$ and $10$.
For target extraction, 
when targets are modeled as M$\delta$-GLMB,
the M$\delta$-GLMB density is first converted to LMB by matching the PHD
and then the tracks with EPs larger than $0.55$ are extracted.
At last,
whenever local filtering and fusion are accomplished,
for M$\delta$-GLMB densities,
label set hypotheses with JEP smaller than $10^{-20}$
and tracks of LMBs with EP  smaller than $10^{-5}$ are discarded.

Two performance indicators will be examined in this section: 
the \textit{optimal subpattern assignment} (OSPA) distance \cite{schuhmacher2008consistent} 
(with order $p=2$ and cutoff $c=50\left[m\right]$)
and the cardinality estimation error.

\subsection{Example 1: DMT with nodes having the same FoV}
Let us first consider a simulation scenario wherein $5$ targets subsequently enter
and then move inside a $5000 \times 5000 \, [m^2]$ surveillance region.
The considered WSN consists of $\left| {\cal N} \right| =10$ sensor nodes deployed at known locations
$ [ {\xi ^i}\;{\eta ^i} ]^\top$ for each $i \in \mathcal{N}$.
The considered scenario is illustrated in Fig. \ref{fig:SCEUniFoV}.

\begin{figure}[tb]
	\centering {
		\begin{tabular}{l}
			\includegraphics[width=0.48\textwidth]{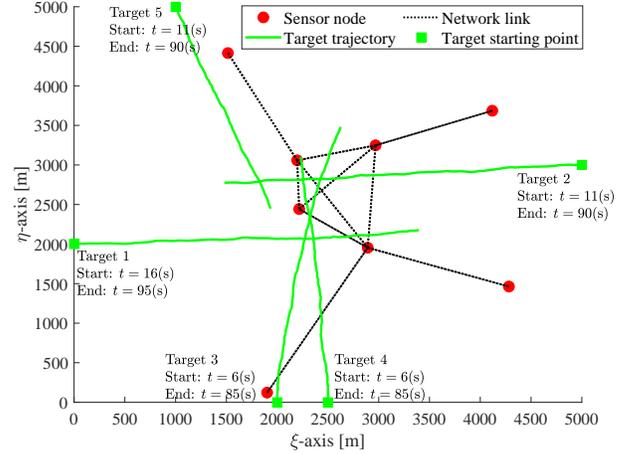}\\
		\end{tabular}
	}
	\caption{Simulated DMT scenario with sensor nodes having the same FoV.}
	\vspace{-0.5\baselineskip}
	\label{fig:SCEUniFoV}
\end{figure}

Now we examine the performance of MIL fusion based on two different probabilities of detection: 
1) $P_d = P_{d,t}^i = 0.98$ and 2) $P_d = P_{d,t}^i = 0.5$
for any time $t$ and sensor node $i\in{\cal N}$.
The number of consensus steps adopted at each node is set to $L=1$.
In order to better illustrate the performance of MIL fusion,
the performance of local trackers without fusion
and of local trackers combined with GCI fusion 
are also considered for comparison.

The average performance over $200$ Monte Carlo trials under different detection probabilities ($P_d=0.98$ and $P_d=0.5$)
are illustrated in Figs. \ref{fig:OSPAUFoV} and \ref{fig:TNUFoV} respectively.
It can be seen that MIL and GCI fusions provide similar results when the detection probability is high.
Conversely, under low detection probability, MIL fusion outperforms GCI fusion
especially for target number estimation.
%Further, it can also be noticed that,
%in the case of low detection probability,
%the performance of MWIG fusion deteriorates whenever the number of consensus steps is increased,
%and  that in this case MWIG fusion performs even worse than no fusion, i.e. local CPHD filtering.
%This is due to the multiplicative nature of the MWIG fusion rule by which
%any missed target detection in a local CPHD filter of a sensor node will cause target disappearance in the fused IIDCP density.
%Consequently, when the detection probability is low and there are more nodes involved in the fusion,
%the probability of occurrence of a missed detection will raise,
%thus negatively affecting DMT performance.
Moreover, it is also observed that
among MIL fusion based algorithms,
the M$\delta$-GLMB based DMT provides better tracking performance 
compared to LMB.
This fact can be seen more clearly in Fig. \ref{fig:OSPAvsPd},
where the average OSPA is reported for different probabilities of detection.
It can also be seen that,
for the M$\delta$-GLMB model,
GCI fusion negatively affects DMT performance 
when $P_d$ decreases below $0.7$ and, similarly,
occurs for the LMB model,
when $P_d$ falls below $0.8$. 

\begin{figure}[tb]
	\centering {
		\begin{tabular}{l}
			\includegraphics[width=0.48\textwidth]{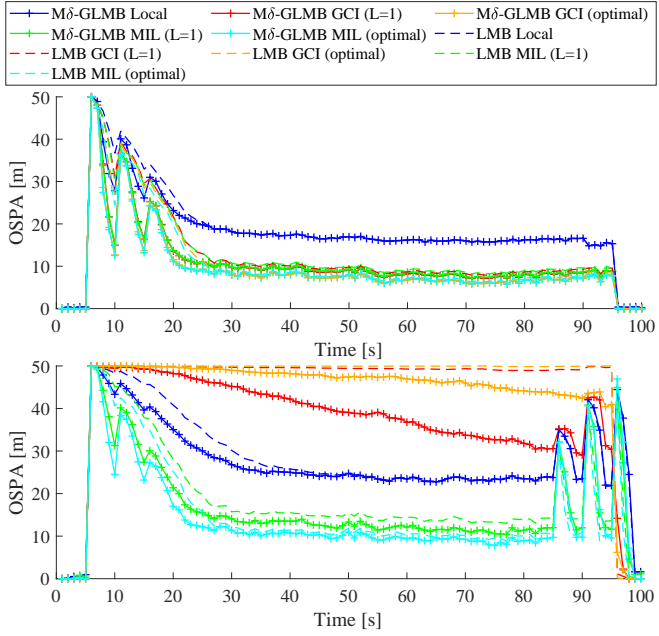}\\
		\end{tabular}
	}
	\caption{OSPA with different detection probabilities, where the top subfigure refers to $P_d=0.98$ and the bottom one to $P_d=0.7$.}
	\vspace{-0.5\baselineskip}
	\label{fig:OSPAUFoV}
\end{figure}

\begin{figure}[tb]
	\centering {
		\begin{tabular}{l}
			\includegraphics[width=0.48\textwidth]{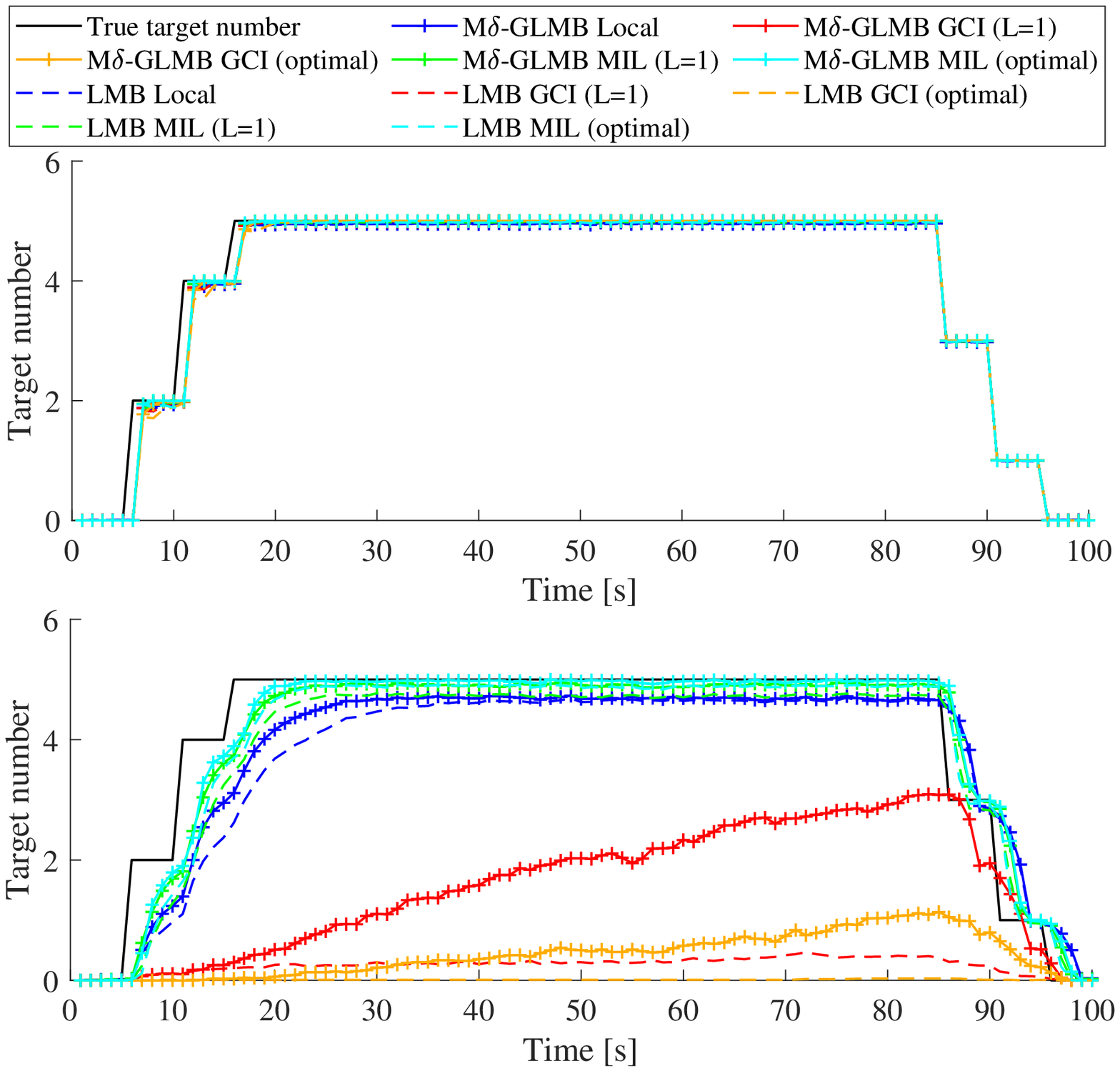}\\
		\end{tabular}
	}
	\caption{Target number estimation with different detection probabilities, where the top subfigure refers to $P_d=0.98$ and the bottom one to $P_d=0.5$.}
	\vspace{-0.5\baselineskip}
	\label{fig:TNUFoV}
\end{figure}

\begin{figure}[tb]
	\centering {
		\begin{tabular}{l}
			\includegraphics[width=0.48\textwidth]{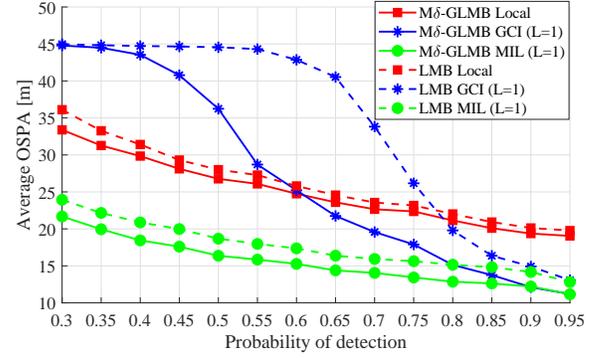}\\
		\end{tabular}
	}
	\caption{Average OSPA under different detection probabilities.}
	\vspace{-0.5\baselineskip}
	\label{fig:OSPAvsPd}
\end{figure}

\subsection{Example 2: DMT with nodes having different FoVs}
Next,
we consider another scenario wherein the trajectories of targets are the same as in Example 1,
while the considered WSN consists of $|{\cal N}| = 4$ nodes.
In this second scenario,
the FoV of each sensor node is taken as a circle  centred at the node location with radius equal to $2500 [m]$.
In order to provide full coverage of the whole surveillance area,
the sensor nodes are regularly placed
as shown in Fig. \ref{fig:SCEDifFoV}. 
As it can be seen,
all targets move to the common FoV of sensor nodes.

\begin{figure}[tb]
	\centering {
		\begin{tabular}{l}
			\includegraphics[width=0.48\textwidth]{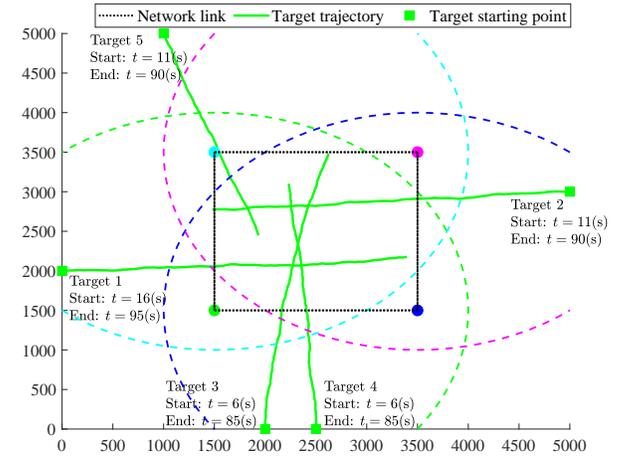}\\
		\end{tabular}
	}
	\caption{Simulated DMT scenario with sensor nodes having different FoVs,
		where the circles with different colors represent different sensor nodes
		and the dashed lines of the same colors delimit the corresponding  FoVs.}
	\vspace{-0.5\baselineskip}
	\label{fig:SCEDifFoV}
\end{figure}

Similar to Example 1,
we also consider both cases of high ($P_d=0.98$) and low ($P_d=0.7$) detection probability
within the FoV of each sensor.
Notice that for each sensor node, we set $P_d=0$ for targets outside the node FoV.
Also in these simulations, the number of consensus steps is set to $L=1$.
The average performance over $200$ Monte Carlo trials under different detection probabilities ($P_d=0.98$ and $P_d=0.7$)
are illustrated in Figs. \ref{fig:OSPADFoV} and \ref{fig:TNDFoV} respectively.
It can be seen that MIL fusion 
is able to detect targets even when they are are in the exclusive FoVs of sensor nodes,
while GCI fusion detects targets only when targets are inside the common FoV of sensor nodes.
Further,
when targets move to the common FoV of sensor nodes,
the same conclusions of Example 1 can be drawn.

\begin{figure}[tb]
	\centering {
		\begin{tabular}{l}
			\includegraphics[width=0.48\textwidth]{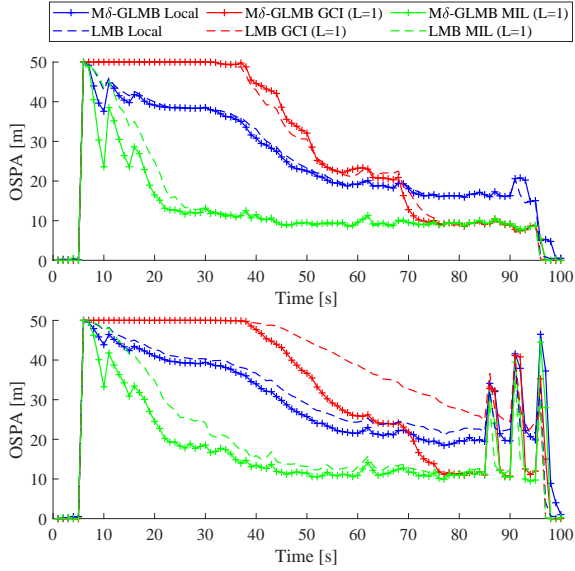}\\
		\end{tabular}
	}
	\caption{OSPA with different detection probabilities, where the top subfigure refers to $P_d=0.98$ and the bottom one to $P_d=0.7$.}
	\vspace{-0.5\baselineskip}
	\label{fig:OSPADFoV}
\end{figure}

\begin{figure}[tb]
	\centering {
		\begin{tabular}{l}
			\includegraphics[width=0.48\textwidth]{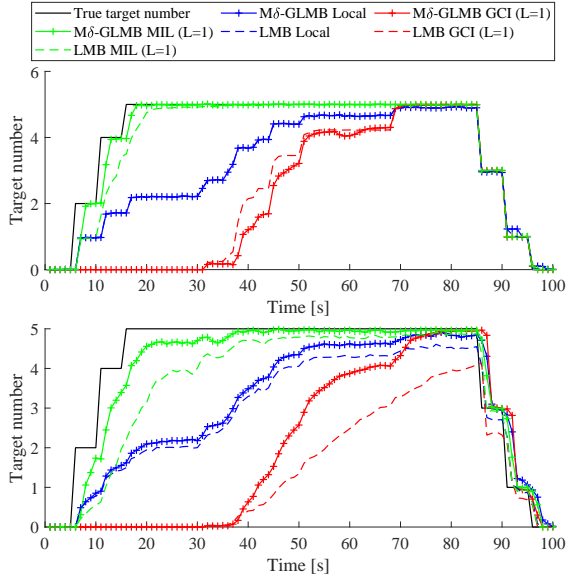}\\
		\end{tabular}
	}
	\caption{Target number estimation with different detection probabilities, where the top subfigure refers to $P_d=0.98$ and the bottom one to $P_d=0.7$.}
	\vspace{-0.5\baselineskip}
	\label{fig:TNDFoV}
\end{figure}

\section{Conclusions}
In this paper,
a new rule which leads to \textit{minimum} (weighted) \textit{information loss} (MIL) is proposed
to handle the problem of fusing \textit{labeled random finite set} (LRFS) densities.
An important property of the proposed fusion rule is that,
combined with the proposed decomposition strategy of LRFS densities,
it can handle the practically relevant case in which local densities are defined in different \textit{fields-of-view} (FoVs).
Further, a strategy is proposed to solve the \textit{label mismatching} (LM) problem among LRFS densities, thus 
strengthening the applicability of the proposed algorithms to real problems.
The performance of the proposed algorithms is assessed by simulation experiments
relative to \textit{distributed multitarget tracking} (DMT)
over a \textit{wireless sensor network} (WSN).
%where sensor nodes have both the same FoV and different FoVs.

\appendices
\section{}  \label{app:P0}
\begin{proof}[Proof of Proposition \ref{pro:P0}]
	From (\ref{eq:GP}), we have
	\begin{align}
	\overline {\boldsymbol \pi} \left( {\bf{X}} \right) &= \sum\limits_{i \in {\cal N}} {{\omega ^i}{{\boldsymbol \pi} ^i}\left( {\bf{X}} \right)}   \nonumber  \\
	& = \sum\limits_{i \in {\cal N}} {{\omega ^i}{p^i}\left( L \right){f^i}\left( {\left. X \right|L} \right)}   \nonumber  \\
	& = \left[ {\sum\limits_{i \in {\cal N}} {{\omega ^i}{p^i}\left( L \right)} } \right] \cdot \left[ {\sum\limits_{i \in {\cal N}} {\frac{{{\omega ^i}{p^i}\left( L \right)}}{{\sum\limits_{j \in {\cal N}} {{\omega ^j}{p^j}\left( L \right)} }}{f^i}\left( {\left. X \right|L} \right)} } \right].
	\end{align}
	Then, the conclusion of Proposition \ref{pro:P0} can be directly obtained.
\end{proof}

\section{}  \label{app:P1}
\begin{proof}[Proof of Proposition \ref{pro:P1}]
	First, it is recalled from (\ref{eq:LRFSN}) that 
	an M$\delta$-GLMB density ${\boldsymbol \pi}_M$ is completely characterized by its JEP $p_M$ and CJPDF $f_M$. 
	Since the aim is to find the optimal M$\delta$-GLMB density according to the MIL criterion, 
	it is straightforward to impose a constraint in the optimization problem of (\ref{eq:MWIL}) as follows
	\begin{align}
	\overline {\boldsymbol \pi}  &= \arg \mathop {\min }\limits_{\boldsymbol \pi}  \sum\limits_{i \in {\cal N}} {{\omega ^i}{D_{\rm KL}}\left( {\left. {{{\boldsymbol \pi}_M ^i}} \right\|{\boldsymbol \pi} } \right)},  \nonumber \\
	& s.t. \quad {\boldsymbol \pi} \left( {\bf X}_n \right) = p\left( L_n \right) \cdot \prod\limits_{k = 1}^n {f_{{l_k}|L_n}\left( {x_k} \right)} ,   \label{eq:LCIMWIL}
	\end{align}
	which amounts to directly looking for the JEP $p_M$ and
	CJPDF $f_{{l_k}|L}$ characterizing the M$\delta$-GLMB density \textcolor{blue}{$\boldsymbol \pi_M$}.
	By resorting to the definition of KLD (\ref{eq:KLD}) and the permutation invariant property of CJPDF (\ref{eq:PI}), we have (\ref{eq:KLDLCILRFS}).
		\begin{figure*}[!t]
		% ensure that we have normalsize text
		\normalsize
		\begin{align}
		{D_{\rm KL}}&\left( {\left. {{{\boldsymbol \pi}_M ^i}} \right\|{\boldsymbol \pi}_M } \right) \nonumber \\
		&= \int {{{\boldsymbol \pi}_M ^i}\left( {\bf{X}} \right)\log \frac{{{{\boldsymbol \pi}_M ^i}\left( {\bf{X}} \right)}}{{{\boldsymbol \pi}_M \left( {\bf{X}} \right)}}\delta {\bf{X}}}    \nonumber  \\
		&  = \sum\limits_{n = 0}^\infty  {\frac{1}{{n!}}\sum\limits_{L_n \in {\cal F}_n(\mathbb L)} {\int {{p^i_M}\left( L_n \right)\prod\limits_{k = 1}^n {f_{{l_k}|L_n}^i\left( {{x_k}} \right)} \log \frac{{{p^i_M}\left( L_n \right)\prod\limits_{k= 1}^n {f_{{l_k}|L_n}^i\left( {{x_k}} \right)} }}{{p_M\left( L_n \right)\prod\limits_{k = 1}^n {{f_{{l_k}|L_n}}\left( {{x_k}} \right)} }}d{x_1}, \ldots ,d{x_n}} } }    \nonumber  \\
		& = \sum\limits_{n = 0}^\infty  {\frac{1}{{n!}}\sum\limits_{L_n \in {\cal F}_n(\mathbb L) } {\int {{p^i_M}\left( L_n \right)\prod\limits_{k = 1}^n {f_{{l_k}|L_n}^i\left( {{x_k}} \right)} \left[ {\log \frac{{{p^i_M}\left( L_n \right)}}{{p_M\left( L_n \right)}} + \sum\limits_{k = 1}^n {\log \frac{{f_{{l_k}|L_n}^i\left( {{x_k}} \right)}}{{{f_{{l_k}|L_n}}\left( {{x_k}} \right)}}} } \right]d{x_1}, \ldots ,d{x_n}} } }   \nonumber  \\
		& = \sum\limits_{n = 0}^\infty  {\frac{1}{{n!}}\sum\limits_{L_n \in {\cal F}_n(\mathbb L) } {{p_M^i}\left( L \right)\log \frac{{{p_M^i}\left( L_n \right)}}{{p_M\left( L_n \right)}}} }  + \sum\limits_{n = 0}^\infty  {\frac{1}{{n!}}\sum\limits_{L_n \subseteq {\mathbb L} } {{p_M^i}\left( L_n \right)\int {\prod\limits_{k = 1}^n {f_{{l_k}|L_n}^i\left( {{x_k}} \right)} \sum\limits_{k = 1}^n {\log \frac{{f_{{l_k}|L_n}^i\left( {{x_k}} \right)}}{{{f_{{l_k}|L_n}}\left( {{x_k}} \right)}}} d{x_1}, \ldots ,d{x_n}} } }    \nonumber  \\
%		&  = \sum\limits_{L \subseteq {\mathbb L} } {{p_M^i}\left( L \right)\log \frac{{{p_M^i}\left( L \right)}}{{p_M\left( L \right)}}}   + \sum\limits_{L \subseteq {\mathbb L}} {{p_M^i}\left( L \right)\sum\limits_{l \in L} {{D_{{\rm{KL}}}}\left( {\left. {f_{l|L}^i} \right\|{f_{l|L}}} \right)} }   \nonumber  \\
		& = {D_{\rm KL}}\left( {\left. {{p_M^i}} \right\|p_M} \right) + \sum\limits_{L  \subseteq {\mathbb L}} {{p^i_M}\left( L \right)\sum\limits_{l \in L} {{D_{\rm KL}}\left( {\left. {f_{l|L}^i} \right\|{f_{l|L}}} \right)} }  \label{eq:KLDLCILRFS}.
		\end{align}
		% IEEE uses as a separator
		\hrulefill
		% The spacer can be tweaked to stop underfull vboxes.
		\vspace*{4pt}
	\end{figure*}

	Then, substituting (\ref{eq:KLDLCILRFS}) into (\ref{eq:LCIMWIL}), we obtain
	\begin{align}
	\overline {\boldsymbol \pi} \hspace{-0.05cm}  &= \arg \mathop {\min }\limits_{{\boldsymbol \pi}_M} \sum\limits_{i \in {\cal N}} {{\omega ^i}{D_{\rm KL}}\left( {\left. {{{\boldsymbol \pi}_M ^i}} \right\|{\boldsymbol \pi}_M } \right)} \nonumber \\
	&= \arg \mathop {\min }\limits_{p_M} \hspace{-0.1cm} \sum\limits_{i \in {\cal N}} {{\omega ^i}{D_{\rm KL}}\left( {\left. {{p^i_M}} \right\|p_M} \right)} \hspace{-0.1cm} + \hspace{-0.1cm} \sum\limits_{L \in  {\mathbb L}} {\sum\limits_{i \in {\cal N}} \hspace{-0.1cm} {\left\{ \hspace{-0.15cm} {\left[ \hspace{-0.05cm} {\sum\limits_{j \in {\cal N}} {{\omega ^j}{p^j_M}\left( L \right)} } \hspace{-0.1cm} \right]} \right.} }  \nonumber  \\
	& \quad  \times \left. {\arg \mathop {\min }\limits_{{f_{l|L}}} \left[ {\frac{{{\omega ^i}{p^i_M}\left( L \right)}}{{\sum\nolimits_{j \in {\cal N}} {{\omega ^j}{p^j_M}\left( L \right)} }}\cdot\sum\limits_{l \in L} {{D_{{\rm{KL}}}}\left( {\left. {f_{l|L}^i} \right\|{f_{l|L}}} \right)} } \right]} \hspace{-0.1cm} \right\} \hspace{-0.1cm} .
	\end{align}
	Finally, applying (\ref{eq:GP}),
	(\ref{eq:Fp}) -- (\ref{eq:Ffl}) can be directly obtained.
\end{proof}

\section{} \label{app:P2}
\begin{proof}[Proof of Proposition \ref{pro:P2}]
	Similarly to (\ref{eq:LCIMWIL}), the fusion problem with respect to multiple LMB densities can be recast into the following optimization problem
	\begin{align}
	\overline {\boldsymbol \pi}  &= \arg \mathop {\min }\limits_{\boldsymbol \pi}  \sum\limits_{i \in {\cal N}} {{\omega ^i}{D_{\rm KL}}\left( {\left. {{\boldsymbol \pi} _\beta ^i} \right\|{\boldsymbol \pi} } \right)} , \nonumber  \\
	& \quad s.t. \; {\boldsymbol \pi} \left( {\bf{X}} \right) = {p_\beta }\left( L \right){f_\beta }\left( {\left. X \right|L} \right),   \label{eq:MWILLMB}
	\end{align}
	where $p_\beta$ and $f_\beta$ are given by (\ref{eq:LMBwl}) and (\ref{eq:LMBspdf}) respectively.
	Specifying the M$\delta$-GLMB densities as LMB densities,  (\ref{eq:KLDLCILRFS}) can be further detailed as
	\begin{align}
	{D_{\rm KL}}\left( {\left. {p_\beta ^i} \right\|p_\beta} \right) &= \sum\limits_{L \subseteq {\mathbb L} } {p_\beta ^i\left( L \right)\log \frac{{\prod\nolimits_{l \in L} {r_l^i}  \cdot \prod\nolimits_{l' \in {\mathbb L}\backslash L} {\left( {1 - r_l^i} \right)} }}{{\prod\nolimits_{l \in L} {{r_l}}  \cdot \prod\nolimits_{l' \in {\mathbb L}\backslash L} {\left( {1 - {r_l}} \right)} }}}   \nonumber  \\
	& = \sum\limits_{L \subseteq {\mathbb L}} {p_\beta ^i\left( L \right)\left[ {\sum\limits_{l \in L} {\log \frac{{r_l^i}}{{{r_l}}}}  + \sum\limits_{l' \in {\mathbb L}\backslash L} {\log \frac{{1 - r_l^i}}{{1 - {r_l}}}} } \right]}   \nonumber  \\
	& = \sum\limits_{l \in {\mathbb L}} \left\{ {\left[ {\sum\limits_{L \subseteq {\mathbb L}\backslash \left\{ l \right\}} {p_\beta ^i\left( {L\bigcup {\left\{ l \right\}} } \right)} } \right]\log \frac{{r_l^i}}{{{r_l}}}}  \right\} \nonumber \\
	& \quad + \sum\limits_{l \in {\mathbb L}} \left\{ {\left[ {\sum\limits_{L \subseteq {\mathbb L}\backslash \left\{ l \right\}} {p_\beta ^i\left( L \right)} } \right]\log \frac{{1 - r_l^i}}{{1 - {r_l}}}} \right\}   \nonumber  \\
	& = \sum\limits_{l \in {\mathbb L}} {\left\{ {\left[ {r_l^i\log \frac{{r_l^i}}{{{r_l}}} + \left( {1 - r_l^i} \right)\log \frac{{\left( {1 - r_l^i} \right)}}{{\left( {1 - {r_l}} \right)}}} \right]} \right\}}   \nonumber  \\
	& = \sum\limits_{l \in {\mathbb L}} {{D_{\rm KL}}\left( {\left. {\rho _l^i} \right\|{\rho _l}} \right)} ,
	\end{align}
	where $\rho_l$ denotes the Bernoulli density with parameter equal to the EP of track $l$, and
	\begin{align}
	\sum\limits_{L \subseteq {\mathbb L} } {{p^i}\left( L \right)\sum\limits_{l \in L} {{D_{\rm KL}}\left( {\left. {f_l^i} \right\|{f_l}} \right)} } = \sum\limits_{l \in {\mathbb L}} {{r^i_l} \cdot {D_{\rm KL}}\left( {\left. {f_l^i} \right\|{f_l}} \right)}.
	\end{align}
	Hence (\ref{eq:MWILLMB}) is re-written as
	\begin{align}
	\overline {\boldsymbol \pi}   &= \sum\limits_{l \in {\mathbb L}} {\arg \mathop {\min }\limits_{{\rho _l}} \sum\limits_{i \in {\cal N}} {{\omega ^i}{D_{{\rm{KL}}}}\left( {\left. {\rho _l^i} \right\|{\rho _l}} \right)} }  + \sum\limits_{l \in {\mathbb L}} {\left[ {\left( {\sum\nolimits_{j \in {\cal N}} {{\omega ^j}r_l^j} } \right)} \right.}   \nonumber  \\
	& \quad \left. { \times \arg \mathop {\min }\limits_{{f_l}} \sum\limits_{i \in {\cal N}} {\frac{{{\omega ^i}r_l^i}}{{\sum\nolimits_{j \in {\cal N}} {{\omega ^j}r_l^j} }}{D_{{\rm{KL}}}}\left( {\left. {f_l^i} \right\|{f_l}} \right)} } \right].
	\end{align}
	Finally, (\ref{eq:FLMBinten}) -- (\ref{eq:FLMBspdf}) can be readily obtained
	by directly applying  (\ref{eq:GP}).
\end{proof}

\section{} \label{app:rem3}
If Proposition \ref{pro:P1} is adopted to fuse LMB densities,
the resulting global density becomes M$\delta$-GLMB 
with JEP $p$ given by (\ref{eq:JEPP2}) and CJPDF $f$ given by
\begin{align}
\overline f\left( {\left. {{X_n}} \right|{L_n}} \right) = \prod\limits_{k = 1}^n {{\overline f_{{l_k}}}\left( {{x_k}} \right)} ,
\end{align}
where
\begin{align}
{\overline f_{{l_k}}}\left( {{x_k}} \right) = \sum\limits_{i \in {\cal N}} {\frac{{{\omega ^i}{p^i}\left( {{L_n}} \right)}}{{\sum\nolimits_{j \in {\cal N}} {{\omega ^j}{p^j}\left( {{L_n}} \right)} }}f_{{l_k}}^i\left( {{x_k}} \right)}.
\end{align}
Following \cite[Section III-B]{reuter2014labeled},
after converting it to LMB by matching the PHD,
the EP $\overline r_l$ of track $l \in {\mathbb L}$ can be computed as
\begin{align}
{\overline r_l} &= \sum\limits_{L \subseteq \mathbb L} {p\left( L \right){{\bf{1}}_L}\left( l \right)}   \nonumber \\
%&= \sum\limits_{L \subseteq \mathbb L\backslash \left\{ l \right\}} {\sum\limits_{i \in {\cal N}} {{\omega ^i}r_l^i\left[ {\prod\limits_{l' \in \mathbb L\backslash \left\{ {L \cup \left\{ l \right\}} \right\}} {\left( {1 - r_{l'}^i} \right)} \prod\limits_{l'' \in L} {r_{l''}^i} } \right]} }   \nonumber  \\
& = \sum\limits_{i \in {\cal N}} {{\omega ^i}r_l^i\varpi \left( {\mathbb L\backslash \left\{ l \right\}} \right)},
\end{align}
where
\begin{align}
\varpi \left(\mathbb L \right) = \sum\limits_{L \subseteq \mathbb L} {\left[ {\prod\limits_{l \in \mathbb L\backslash L} {\left( {1 - r_l^i} \right)} \prod\limits_{l' \in L} {r_{l'}^i} } \right]} 
\end{align}
amounts to computing the summation of JEPs of the LMB density defined over all possible label subsets of $\mathbb L$, 
thus always equals one, i.e. $\varpi \left( \mathbb L \right) \equiv 1$.
Then, it is straightforward to see that ${\overline r_l} = \sum\nolimits_{i \in {\cal N}} {{\omega ^i}r_l^i}$.
Further, the fused PDF $\overline f_l$ of track $l \in {\mathbb L}$ is computed as
\begin{align}
{\overline f_l}\left( x \right) &= \frac{{\sum\limits_{L \subseteq \mathbb L} {\overline p\left( L \right) \cdot {\overline f_{\left. l \right|L}}\left( x \right){{\bf{1}}_L}\left( l \right)} }}{{\sum\limits_{L \subseteq \mathbb L} {\overline p\left( L \right){{\bf{1}}_L}\left( l \right)} }}  \nonumber  \\
& = \frac{{\sum\limits_{L \subseteq \mathbb L} {\sum\limits_{j \in {\cal N}} {{\omega ^j}{p^j}\left( {L} \right)} \sum\limits_{i \in {\cal N}} {\frac{{{\omega ^i}{p^i}\left( L \right)}}{{\sum\limits_{j \in {\cal N}} {{\omega ^j}{p^j}\left( L \right)} }}f_l^i\left( x \right)} } {{\bf{1}}_L}\left( l \right)}}{{\sum\limits_{L \subseteq \mathbb L} {\sum\limits_{i \in {\cal N}} {{\omega ^i}{p^i}\left( L \right)} {{\bf{1}}_{L'}}\left( l \right)} }}  \nonumber  \\
& = \frac{{\sum\limits_{i \in {\cal N}} {{\omega ^i}f_l^i\left( x \right)\sum\limits_{L \subseteq \mathbb L} {{p^i}\left( L \right){{\bf{1}}_L}\left( l \right)} } }}{{\sum\limits_{i \in {\cal N}} {{\omega ^i}r_l^i} }}   \nonumber  \\
& = \frac{{\sum\limits_{i \in {\cal N}} {{\omega ^i}r_l^jf_l^i\left( x \right)} }}{{\sum\nolimits_{j \in {\cal N}} {{\omega ^j}r_l^j} }},
\end{align}
where ${\bf{1}}_L$ represents the inclusion function defined as
\begin{align}
{{\bf{1}}_L}\left( l \right) = \left\{ \begin{array}{l}
1,\quad {\rm if}\; l \in L\\
0,\quad {\rm if}\; l \notin L
\end{array} \right. .
\end{align}
Then, it can be seen immediately that the converted LMB density from the fused M$\delta$-GLMB density 
is the one obtained by using the results of Proposition \ref{pro:P2}.

\section{} \label{app:the2}
\textit{Proof of Theorem \ref{the:T2}.}
The KLD from the fused M$\delta$-GLMB density of Proposition \ref{pro:P1} to the MIL-OFD 
cannot be directly computed. 
However, it turns out that it is bounded by \cite[eq. (4)]{do2003fast}
\begin{align}
&{D_{\rm KL}}( {\sum\limits_{i \in {\cal N}} {{\omega ^i}\pi _M^i} ||{{\bar \pi }_M}} ) \nonumber  \\
&= {D_{\rm KL}}( {\sum\limits_{i \in {\cal N}} {{\omega ^i}\pi _M^i} ||\sum\limits_{i \in {\cal N}} {{\omega ^i}{{\bar \pi }_M}} } )   \nonumber  \\
&\le \sum\limits_{i \in {\cal N}} {{\omega ^i}{D_{\rm KL}}\left( {\pi _M^i||{{\bar \pi }_M}} \right)}  \nonumber \\
&= \sum\limits_{i \in {\cal N}} {{\omega ^i}[ {{D_{\rm KL}}\left( {p_M^i||{{\bar p}_M}} \right) \hspace{-0.1cm} + \hspace{-0.1cm}  \sum\limits_{L \subseteq {\mathbb L}} {p_M^i\left( L \right)\sum\limits_{l \in L} {{D_{\rm KL}}( {f_{l|L}^i||{{\bar f}_{l|L}}} )} } } ]}   \nonumber  \\
& = \sum\limits_{i \in {\cal N}} {{\omega ^i}[{D_{\rm KL}}(p_M^i||\sum\limits_{j \in {\cal N}} {{\omega ^j}p_M^j} )} \nonumber  \\
& \quad  + \sum\limits_{L \subseteq {\mathbb L}} {p_M^i\left( L \right)\sum\limits_{l \in L} {{D_{\rm KL}}(f_{l|L}^i||\sum\limits_{j \in {\cal N}} {{\omega ^j}f_{l|L}^j} )} } ]   \nonumber  \\
& \le \sum\limits_{i \in {\cal N}} {{\omega ^i}[\sum\limits_{j \in {\cal N}} {{\omega ^j}{D_{\rm KL}}( {p_M^i||p_M^j} )} } \nonumber  \\
& \quad  + \sum\limits_{j \in {\cal N}} {{\omega ^j}\sum\limits_{L \subseteq {\mathbb L}} {p_M^i\left( L \right)\sum\limits_{l \in L} {{D_{\rm KL}}( {f_{l|L}^i||f_{l|L}^j} )} } } ]  \nonumber  \\
& = \sum\limits_{i \in {\cal N}} {\sum\limits_{j \in {\cal N},j \ne i} {{\omega ^i}{\omega ^j}{D_{\rm KL}}( {\left. {\pi _M^i} \right\|\pi _M^j} )} } ,  \label{eq:PUB}
\end{align}
which proves (\ref{eq:KLDMO}) in Theorem \ref{the:T2}.
Furthermore, the proof of  (\ref{eq:KLDLO}) can be accomplished 
directly following the steps of (\ref{eq:PUB}),
and is therefore omitted. \qed

\section{} \label{app:P3}
\begin{proof}[Proof of Proposition \ref{Pro:FoID}]
	Given an LRFS 
	\[{\bf{X}} = \left\{ {\left( {{x_1},{l_1}} \right), \ldots ,\left( {{x_n},{l_n}} \right)} \right\}\]
	and the disjoint label spaces $\mathbb L_1,\ldots,\mathbb L_M$,
	we denote 
	\begin{align*}
	{{\bf{X}}_m} &= \left\{ {{\bf{X}}':{\cal L}\left( {{\bf{X}}'} \right) = {\cal L}\left( {\bf{X}} \right)\bigcap {{L_m}} } \right\} \\
	& \buildrel \Delta \over = \left\{ {\left( {{x_{1,m}},{l_{1,m}}} \right), \ldots ,\left( {{x_{{n_m},m}},{l_{{n_m},m}}} \right)} \right\}.
	\end{align*}
	
%	The global density ${\boldsymbol \pi}$ leading to MWIL can be computed as
%	\begin{align}
%	{\boldsymbol \pi} = \arg \mathop {\min }\limits_{\overline{\boldsymbol \pi}} \sum\limits_{i \in N} {{\omega ^i}\cdot{D_{{\rm{KL}}}}\left( {\left. {\prod\limits_{m = 1}^M {{\boldsymbol \pi}_m^i} } \right\|\overline{\boldsymbol \pi}} \right)}.   \label{eq:MWILInd}
%	\end{align}
	
	Substituting the definition of KLD into (\ref{eq:MWILInd}) and recalling (\ref{eq:Prop1})-(\ref{eq:Prop2}), we have (\ref{eq:MWILINDKLD}). 
	Then, (\ref{eq:MWILInd}) can be readily obtained by exploiting the results of Proposition \ref{pro:P0}.
	
	\begin{figure*}[!t]
		% ensure that we have normalsize text
		\normalsize
		\begin{align}
		 \overline {\boldsymbol \pi}  &= \arg \mathop {\min }\limits_{\left\{ {{ {\boldsymbol \pi} _m}} \right\}_{m = 1}^M} \sum\limits_{i \in {\cal N}} {{\omega ^i}{D_{\rm KL}}\left( {\left. {\prod\limits_{m = 1}^M {{\boldsymbol \pi} _m^i} } \right\|\prod\limits_{m = 1}^M {{ {\boldsymbol \pi} _m}} } \right)}    \nonumber  \\
		& = \arg \mathop {\min }\limits_{\left\{ {{ {\boldsymbol \pi} _m}} \right\}_{m = 1}^M} \sum\limits_{i \in {\cal N}} {{\omega ^i} \cdot \sum\limits_{n = 0}^\infty  {\sum\limits_{L \in {\cal F}_n(\mathbb L)} {\int { \ldots \int {\prod\limits_{m = 1}^M {{\boldsymbol \pi} _m^i\left( {{{\bf{X}}_m}} \right)} \sum\limits_{m' = 1}^M {\log \frac{{{\boldsymbol \pi} _{m'}^i\left( {{{\bf{X}}_{m'}}} \right)}}{{{ {\boldsymbol \pi} _{m'}}\left( {{{\bf{X}}_{m'}}} \right)}}} d{x_1}, \ldots ,d{x_n}} } } } }   \nonumber  \\
		& = \arg \mathop {\min }\limits_{\left\{ {{ {\boldsymbol \pi} _m}} \right\}_{m = 1}^M} \sum\limits_{m' = 1}^M {\sum\limits_{i \in {\cal N}} {{\omega ^i} \cdot \sum\limits_{n = 0}^\infty  {\sum\limits_{L \in {\cal F}_n(\mathbb L)} {\int { \cdots \int {\prod\limits_{m = 1}^M {{\boldsymbol \pi} _m^i\left( {{{\bf{X}}_m}} \right)} \log \frac{{{\boldsymbol \pi} _{m'}^i\left( {{{\bf{X}}_{m'}}} \right)}}{{{ {\boldsymbol \pi} _{m'}}\left( {{{\bf{X}}_{m'}}} \right)}}d{x_1}, \ldots ,d{x_n}} } } } } }    \nonumber  \\
		& = \arg \mathop {\min }\limits_{\left\{ {{ {\boldsymbol \pi} _m}} \right\}_{m = 1}^M} \sum\limits_{m' = 1}^M {\sum\limits_{i \in {\cal N}} {{\omega ^i} \cdot \sum\limits_{n = 0}^\infty  {\sum\limits_{L \in {\cal F}_n(\mathbb L)} {\int { \cdots \int {\prod\limits_{m = 1,m \ne m'}^M {{\boldsymbol \pi} _m^i\left( {{{\bf{X}}_m}} \right)d{x_{1,m}}, \ldots ,d{x_{{n_m},m}}} } } } } } }  \nonumber  \\
		& \quad\quad\quad\quad\quad\quad\quad  \times \int { \cdots \int {{\boldsymbol \pi} _{m'}^i\left( {{{\bf{X}}_{m'}}} \right)\log \frac{{{\boldsymbol \pi} _{m'}^i\left( {{{\bf{X}}_{m'}}} \right)}}{{{ {\boldsymbol \pi} _{m'}}\left( {{{\bf{X}}_{m'}}} \right)}}d{x_{1,m'}}, \ldots ,d{x_{{n_{m'}},m'}}} }   \nonumber  \\
		& = \sum\limits_{m = 1}^M {\arg \mathop {\min }\limits_{{{\boldsymbol \pi} _m}} \sum\limits_{i \in {\cal N}} {{\omega ^i} \cdot {D_{\rm KL}}\left( {\left. {{\boldsymbol \pi} _m^i} \right\|{{\boldsymbol \pi} _m}} \right)} }
		\label{eq:MWILINDKLD}.
		\end{align}
		% IEEE uses as a separator
		\hrulefill
		% The spacer can be tweaked to stop underfull vboxes.
		\vspace*{4pt}
	\end{figure*}	
\end{proof}

%\section{} \label{app:P4}
%\begin{proof}[Proof of Theorem \ref{the:2}]
%	The global density ${\boldsymbol \pi}$ leading to MWIG can be computed as
%	\begin{align}
%	\overline{\boldsymbol \pi}'  = \arg \mathop {\min }\limits_{\left\{ {{\overline{\boldsymbol \pi} '_m}} \right\}_{m = 1}^M} \sum\limits_{i \in {\cal N}} {{\omega ^i}\cdot {D_{\rm KL}}\left( {\left. \prod\limits_{m = 1}^M {{\overline{\boldsymbol \pi}' _m}}  \right\| {\prod\limits_{m = 1}^M {{\boldsymbol \pi} _m^i} }  } \right)} .   \label{eq:MWIGInd}
%	\end{align}
%	Notice that the difference between (\ref{eq:MWIGInd}) and (\ref{eq:MWILInd})
%	lies only at the exchange of arguments.
%	Then by implementing a similar derivation procedure to (\ref{eq:MWILINDKLD}),
%	formula (\ref{eq:MWIGInd}) can be re-written as
%	\begin{align}
%	\overline{\boldsymbol \pi}'  = \sum\limits_{m = 1}^M {\arg \mathop {\min }\limits_{{\overline{\boldsymbol \pi}' _m}} \sum\limits_{i \in {\cal N}} {{\omega ^i} \cdot {D_{\rm KL}}\left( {\left. {\overline{\boldsymbol \pi}' _m} \right\| {{\boldsymbol \pi} _m^i} } \right)} }.
%	\end{align}
%	Then the result of Theorem \ref{the:2} can be immediately seen.
%\end{proof}

\section{} \label{app:P5}
\begin{proof}[Proof of Proposition \ref{pro:dec}]
	The purpose is to find $M$ mutually independent sub-densities ${\boldsymbol \pi _m} = \left( {{p_m},{f_m}} \right)$ 
	defined in $M$ disjoint label spaces
	such that their product minimizes the KL divergence.
	By definition, we have (\ref{eq:LRFSdec}),
	where $C$ denotes the constant that is not related to ${\boldsymbol \pi _m}$, $m = 1,\ldots,M$.
	Then. by exploiting the results of Proposition \ref{pro:P0} and minimizing ${D_{\rm KL}}\left( {\left. \boldsymbol \pi  \right|\prod\nolimits_{m = 1}^M {{\boldsymbol \pi _m}} } \right)$,
	the conclusion of Proposition \ref{pro:dec} can be proved.
	\begin{figure*}[!t]
		% ensure that we have normalsize text
		\normalsize
		\begin{align}
		& {D_{\rm KL}}\left( {\left. \boldsymbol \pi  \right|\prod\limits_{m = 1}^M {{\boldsymbol \pi _m}} } \right) \nonumber \\
		&\quad= \sum\limits_{n = 0}^\infty  {\sum\limits_{L \in {\cal F}_n(\mathbb L)} {\int { \cdots \int {\boldsymbol \pi \left( {\bf{X}} \right)\log \frac{{\boldsymbol \pi \left( {\bf{X}} \right)}}{{\prod\limits_{m = 1}^M {{\boldsymbol \pi _m}\left( {{{\bf{X}}_m}} \right)} }}d{x_1}, \ldots ,d{x_n}} } } }    \nonumber  \\
		&\quad = \sum\limits_{n = 0}^\infty  {\sum\limits_{L \in {\cal F}_n(\mathbb L)} {\int { \cdots \int {p\left( L \right)f\left( {\left. X \right|L} \right)\log \frac{{p\left( L \right)f\left( {\left. X \right|L} \right)}}{{\prod\limits_{m = 1}^M {{p_m}\left( {{L_m}} \right){f_m}\left( {\left. {{X_m}} \right|{L_m}} \right)} }}d{x_1}, \ldots ,d{x_n}} } } }     \nonumber  \\
		&\quad= \sum\limits_{n = 0}^\infty  {\sum\limits_{L \in {\cal F}_n(\mathbb L)} {p\left( L \right)\log \frac{{p\left( L \right)}}{{\prod\limits_{m = 1}^M {{p_m}\left( {{L_m}} \right)} }}} }  + \sum\limits_{n = 0}^\infty  {\sum\limits_{L \in {\cal F}_n(\mathbb L)} {p\left( L \right)\int { \cdots \int {f\left( {\left. X \right|L} \right)\log \frac{{f\left( {\left. X \right|L} \right)}}{{\prod\limits_{m = 1}^M {{f_m}\left( {\left. {{X_m}} \right|{L_m}} \right)} }}d{x_1}, \ldots ,d{x_n}} } } }     \nonumber  \\
		&\quad = \sum\limits_{m = 1}^M {\sum\limits_{n = 0}^\infty  {\sum\limits_{L \in {\cal F}_n(\mathbb L)} {p\left( L \right)\log \frac{{p\left( L \right)}}{{{p_m}\left( {{L_m}} \right)}}} } }  + \sum\limits_{m = 1}^M {\sum\limits_{n = 0}^\infty  {\sum\limits_{L \in {\cal F}_n(\mathbb L)} {p\left( L \right)\int { \cdots \int {f\left( {\left. X \right|L} \right)\log \frac{{f\left( {\left. X \right|L} \right)}}{{{f_m}\left( {\left. {{X_m}} \right|{L_m}} \right)}}d{x_1}, \ldots ,d{x_n}} } } } }     \nonumber  \\
		&\quad= C + \sum\limits_{m = 1}^M {\left[ {\sum\limits_{L:L\bigcap {{\mathbb L_m}}  = {L_m}} {p\left( L \right)} } \right]\log \frac{{\sum\limits_{L:L\bigcap {{\mathbb L_m}}  = {L_m}} {p\left( L \right)} }}{{{p_m}\left( {{L_m}} \right)}}}  + \sum\limits_{m = 1}^M {\sum\limits_{L:L\bigcap {{\mathbb L_m}}  = {L_m}} {\tilde \omega \left( L \right){D_{\rm KL}}\left( {\left. {{{\tilde f}_m}} \right\|{f_m}} \right)} } 
		\label{eq:LRFSdec}.
		\end{align}
		% IEEE uses as a separator
		\hrulefill
		% The spacer can be tweaked to stop underfull vboxes.
		\vspace*{4pt}
	\end{figure*}	
\end{proof}

%\section*{Acknowledgment}


\ifCLASSOPTIONcaptionsoff
  \newpage
\fi



%\bibliographystyle{IEEEtran}

%\bibliographystyle{IEEEtran}
%\bibliography{IEEEabrv,reference}


\end{document}